\documentclass[11pt]{article}

\usepackage[utf8]{inputenc}
\usepackage[T1]{fontenc}
\usepackage{amsmath,amssymb,amsthm,mathtools}
\usepackage{enumerate}
\usepackage{booktabs}
\usepackage[colorlinks=true,linkcolor=blue,citecolor=blue]{hyperref}
\usepackage{cleveref}
\usepackage[margin=1.2in]{geometry}
\usepackage{authblk}
\usepackage{orcidlink}
\usepackage{tikz}
\usetikzlibrary{calc, shapes.geometric, arrows.meta, positioning, decorations.pathmorphing}
\usepackage{caption}
\theoremstyle{plain}
\newtheorem{theorem}{Theorem}[section]
\newtheorem{proposition}[theorem]{Proposition}
\newtheorem{lemma}[theorem]{Lemma}
\newtheorem{corollary}[theorem]{Corollary}

\theoremstyle{definition}
\newtheorem{definition}[theorem]{Definition}
\newtheorem{construction}[theorem]{Construction}

\theoremstyle{remark}
\newtheorem{remark}[theorem]{Remark}

\newcommand{\R}{\mathbb{R}}
\newcommand{\Z}{\mathbb{Z}}
\newcommand{\Jplus}{J^+}
\newcommand{\Jminus}{J^-}
\newcommand{\Vol}{\mathrm{Vol}}

\DeclareMathOperator*{\argmin}{arg\,min}

\begin{document}

\title{On the Uniqueness of Embeddings\\ of Causal Sets}
\author[1,2]{Nathan Madsen\,\orcidlink{0000-0001-7791-0481}\thanks{Email: \texttt{madsen@ucsb.edu}}}
\affil[1]{Department of Physics, University of California, Santa Barbara}
\affil[2]{Department of Mathematics, University of California, Santa Barbara}
\date{June 28, 2026}
\maketitle

\begin{abstract}
	We introduce the notion of a well-conditioned embedding of a causal set into a Lorentzian manifold and prove that if a causal set admits well-conditioned embeddings into two manifolds, then their interiors are related by an $\varepsilon$-approximate isometry. To justify the definition, we show that in the high-density limit a Poisson sprinkling almost surely yields a causal set possessing a well-conditioned embedding. The error $\varepsilon$ is given explicitly and tends to zero in the high-density limit (Corollary~\ref{cor:quantitative}).
\end{abstract}

\section{Introduction}\label{sec:intro}

A causal set is a locally finite partial order, conjectured to encode the geometry of a Lorentzian manifold in the case that it faithfully embeds. The \emph{Hauptvermutung} of causal set theory asserts that this embedding is essentially unique: a single causal set cannot embed faithfully into two macroscopically distinct spacetimes (Figure~\ref{fig:embedding}). We prove a quantitative form of this conjecture for finite causal sets arising from Poisson sprinkling.

Recently, M\"uller~\cite{Mueller2025} formulated two mathematically precise versions of the conjecture and showed that the direct ``finite-set'' formulation, in which the embeddings preserve only the causal order, is false: he constructs explicit causal sets faithfully embeddable into pairs of non-isometric Lorentzian manifolds. M\"uller's countable version, by contrast, requires the embedding to be a \emph{dense} subset of the manifold and works in a limit where the causal set has been pre-completed to a continuum object. Neither result addresses the physically relevant CST setting: causal sets in CST are \emph{finite} structures obtained by Poisson sprinkling at finite density~\cite{Surya2019}, in which the causal order alone is famously insufficient (this is essentially the content of M\"uller's negative result), but the combination of order \emph{and} volume information~(F2), together with the asymptotic high-density limit, is expected to suffice. This article addresses precisely that setting.

Our argument isolates the geometric content of the conjecture from its probabilistic content. We introduce a deterministic condition on an embedding $f\colon C\to(M,g)$, the \emph{well-conditioned embedding} (Definition~\ref{def:faithful}), which augments the classical ``order $+$ number'' faithfulness of Bombelli--Lee--Meyer--Sorkin~\cite{Bombelli1987} (order-preservation and volume-faithfulness) with a longest-chain/proper-time correspondence. The proof then divides into two independent parts.

\textbf{Part I} is pure geometry: if a finite causal set admits a well-conditioned embedding into each of two spacetimes, then their interiors are related by an $\varepsilon$-approximate isometry with error $\varepsilon \to 0$ (made precise in Theorem~\ref{thm:haupt}, Eq.~\eqref{eq:eps-final}) as the conditioning sharpens. This section is wholly deterministic.

\textbf{Part II} is a probabilistic result: a Poisson sprinkling at density $\rho$ yields a well-conditioned embedding with probability $1-(\rho V_M)^{-K'}$, in the high-density limit. The classical conditions follow from Chernoff concentration; the chain correspondence from Bollob\'as-Brightwell; the anchor scaffold then follows deterministically from volume-faithfulness (Lemma~\ref{lem:scaffold}), with no separate probabilistic estimate. 

Together, Parts~I and~II reduce the Hauptvermutung for Poisson sprinklings to a deterministic geometric statement: a single causal set cannot be a faithful sprinkling of two macroscopically distinct spacetimes (Corollary~\ref{cor:quantitative}).

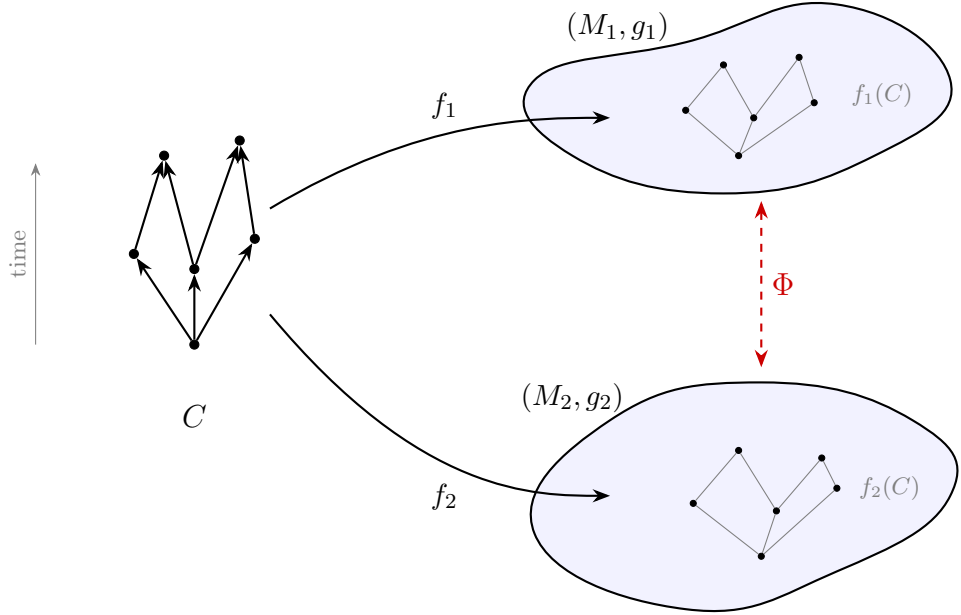
\begin{figure}[h!]
    \centering
    \begin{tikzpicture}[
        >=Stealth,
        node distance=1.5cm,
        causet node/.style={circle, draw=black, fill=black, inner sep=1.2pt},
        emb node/.style={circle, draw=black, fill=black, inner sep=0.8pt},
        emb link/.style={gray, thin}
    ]

    \begin{scope}[shift={(-0.5, 0)}]
        \node[causet node] (c0) at (0, 0) {};
        \node[causet node] (c1) at (-0.8, 1.2) {};
        \node[causet node] (c2) at (0, 1) {};
        \node[causet node] (c3) at (0.8, 1.4) {};
        \node[causet node] (c4) at (-0.4, 2.5) {};
        \node[causet node] (c5) at (0.6, 2.7) {};

        \draw[->, thick] (c0) -- (c1);
        \draw[->, thick] (c0) -- (c2);
        \draw[->, thick] (c0) -- (c3);
        \draw[->, thick] (c1) -- (c4);
        \draw[->, thick] (c2) -- (c4);
        \draw[->, thick] (c2) -- (c5);
        \draw[->, thick] (c3) -- (c5);

        \node[below=0.6cm of c0, font=\large] {$C$};
    \end{scope}

    \draw[->, gray] (-2.6, 0.0) -- (-2.6, 2.4)
        node[midway, left, gray, font=\footnotesize, rotate=90, anchor=south] {time};

    \begin{scope}[shift={(7, 3)}]
        \draw[thick, fill=blue!5] plot [smooth cycle, tension=0.8] coordinates {
            (-2.5,-0.5) (-0.5, -1) (1.5,-0.5) (2.5, 0.5) (1, 1.5) (-1, 1) (-3, 0.5)
        };
        \node at (-1.9, 1.2) {$(M_1, g_1)$};

        \node[emb node] (e0) at (-0.3,-0.5) {};
        \node[emb node] (e1) at (-1.0, 0.1) {};
        \node[emb node] (e2) at (-0.1, 0.0) {};
        \node[emb node] (e3) at ( 0.7, 0.2) {};
        \node[emb node] (e4) at (-0.5, 0.7) {};
        \node[emb node] (e5) at ( 0.5, 0.8) {};
        \draw[emb link] (e0)--(e1) (e0)--(e2) (e0)--(e3)
                        (e1)--(e4) (e2)--(e4) (e2)--(e5) (e3)--(e5);
        \node[gray, font=\footnotesize] at (1.6,0.3) {$f_1(C)$};
    \end{scope}

    \begin{scope}[shift={(7, -2)}]
        \draw[thick, fill=blue!5] plot [smooth cycle, tension=0.8] coordinates {
            (-3,-0.5) (-1, -1.5) (1,-1) (2.5, 0) (2, 1) (0, 1.5) (-2, 1)
        };
        \node at (-2.5, 1.3) {$(M_2, g_2)$};

        \node[emb node] (g0) at ( 0.0,-0.8) {};
        \node[emb node] (g1) at (-0.9,-0.1) {};
        \node[emb node] (g2) at ( 0.2,-0.2) {};
        \node[emb node] (g3) at ( 1.0, 0.1) {};
        \node[emb node] (g4) at (-0.3, 0.6) {};
        \node[emb node] (g5) at ( 0.8, 0.5) {};
        \draw[emb link] (g0)--(g1) (g0)--(g2) (g0)--(g3)
                        (g1)--(g4) (g2)--(g4) (g2)--(g5) (g3)--(g5);
        \node[gray, font=\footnotesize] at (1.7,0.1) {$f_2(C)$};
    \end{scope}

    \draw[->, thick] (0.5, 1.8) .. controls (2.5, 3) and (4, 3) .. (5, 3)
        node[midway, above left] {$f_1$};

    \draw[->, thick] (0.5, 0.4) .. controls (2.5, -2) and (4, -2) .. (5, -2)
        node[midway, below left] {$f_2$};

    \draw[<->, dashed, thick, color=red!80!black] (7, 1.9) -- (7, -0.3)
        node[midway, right, font=\large] {$\Phi$};

    \end{tikzpicture}
    \caption{The central problem. A discrete causal set $C$ (a locally finite partial order, with time running upward) faithfully embeds into two spacetimes $(M_1, g_1)$ and $(M_2, g_2)$ as the point sets $f_1(C)$ and $f_2(C)$. The theorem produces an approximate isometry $\Phi$ relating them, $\Phi\circ f_1\approx f_2$.}
    \label{fig:embedding}
\end{figure}

\section{Definitions and setup}\label{sec:setup}

\subsection{Getting started}

\begin{definition}[Causal set]\label{def:causet}
A \emph{causal set} $(C,\preceq)$ is a locally finite partially ordered set: for every $x,y\in C$ with $x\preceq y$, the interval $I(x,y):=\{z\in C:x\preceq z\preceq y\}$ is finite.
\end{definition}

\begin{definition}[Globally hyperbolic spacetime]
\label{def:globally-hyp}
A Lorentzian manifold $(M,g)$ of dimension $d$ is \emph{globally hyperbolic} if it admits a Cauchy surface and every causal diamond $\Jplus(p)\cap\Jminus(q)$ is compact.
\end{definition}

\begin{construction}[Global auxiliary Riemannian metric]
\label{con:aux-riem}
A globally hyperbolic $(M,g)$ admits a smooth Cauchy time function $t\colon M\to\R$~\cite{BernalSanchez2005}. Set $T:=-\nabla t/|\nabla t|_g$ and define $h:=g+2T^\flat\otimes T^\flat$. Then $h$ is a smooth, globally defined, positive-definite metric with $h=\delta+O(L^2/\lambda^2)$ in $h$-normal coordinates within any region of diameter $L\ll\lambda$. We write $h_1,h_2$ for the auxiliary metrics on $M_1,M_2$.
\end{construction}

\begin{remark}[Independence of auxiliary foliation choices]
\label{rem:cauchy-choice}
While the auxiliary metrics $h_1$ and $h_2$ depend strictly on the choices of Cauchy time functions on $M_1$ and $M_2$, these two foliations are chosen completely independently of one another; no compatibility between them is required. The Karcher mean (the moving center of mass defining $\Phi$, introduced in \S\ref{subsec:karcher}; here $\ell$ denotes the smoothing scale of \S\ref{subsec:scales}) uses $h_1$ only for the source weights $w_k(x)=\exp(-d_{h_1}^2(x,p_k)/\ell^2)$ and $h_2$ only for the target objective $\sum w_k\,d_{h_2}^2(y,q_k)$. Within any mesoscopic region, any such auxiliary metric satisfies $h_i=\delta+O(L^2/\lambda^2)$ in its own normal coordinates, regardless of the specific choice of $T$. The final conclusion that $\Phi$ is an approximate isometry between $(M_1,g_1)$ and $(M_2,g_2)$ is a statement about the intrinsic Lorentzian geometry, and is entirely invariant under the arbitrary choice of auxiliary foliations.
\end{remark}

\begin{remark}[Adapted normal coordinates]
\label{rem:adapted-coords}
For computations involving the Lorentzian structure, we use $h_i$-normal coordinates with the orthonormal basis chosen so that $e_0=T_i$. In such adapted coordinates, $h_i=\delta+O(|p|^2/\lambda^2)$ (since they are $h_i$-normal) \emph{and} $g_i=\eta+O(|p|^2/\lambda^2)$ (since $T_i^\flat=-dx^0$ in this basis, so $g_i(0)=h_i(0)-2T_iT_i^T=\delta-2\,e_0e_0^T=\eta$). This dual property (a single coordinate chart in which both the auxiliary Riemannian metric is Euclidean and the Lorentzian metric is Minkowski to leading order) is what allows us to apply both Riemannian Karcher-mean estimates and Lorentzian Alexandrov--Zeeman arguments at the same basepoint.
\end{remark}

\subsection{Well-Conditioned Embeddings}

\begin{definition}[Well-conditioned embedding]
\label{def:faithful}
Let $C$ be a finite causal set and $(M,g)$ a globally  hyperbolic Lorentzian manifold of dimension $d$.  Define the \emph{intrinsic curvature scale}
\begin{equation}\label{eq:lambda-intrinsic}
   \lambda := \sup_{T}\min\!\bigl(\mathrm{inj}(M,h_T),\ |\mathrm{Rm}[g]|^{-1/2}\bigr),
\end{equation}
where the supremum runs over all smooth timelike unit vector fields $T$ on $(M,g)$, and $h_T:=g+2T^\flat\otimes T^\flat$ is the auxiliary  Riemannian metric of Construction~\ref{con:aux-riem} associated with $T$.  The supremum is achieved (or approached) by a $T$ adapted to a smooth Bernal--S\'anchez Cauchy time function~\cite{BernalSanchez2005}, and for globally hyperbolic spacetimes with bounded geometry, $\lambda$ is comparable to the purely Lorentzian quantity $|\mathrm{Rm}[g]|^{-1/2}$ up to a dimensional constant.

Fix parameters $\rho>0$ (density), $K_d>0$ (tolerance constant), and a dimensional constant $c_*\in(0,1)$ defining the admissible mesoscopic range.

We define a map $f\colon C\to M$ as a $(\rho;c_*,K_d)$-\emph{well-conditioned embedding} if:

\begin{enumerate}[(F1)]
\item \textbf{Order-preservation (exact):}
For all $x,y\in C$, $f(x)\in J^-_g(f(y))$ iff $x\preceq y$.

\item \textbf{Scale-dependent uniform density:}
There exists a non-empty range $[\tau_{\min}, c_* \lambda]$ such that every causal diamond $D=J^+_g(p)\cap J^-_g(q)\subset M$ with proper time
\begin{equation}\label{eq:F2-range}
   \tau_g(p,q)\in[\tau_{\min},c_*\lambda],
   \qquad
   \tau_{\min}:=c_*^{-1}\rho^{-1/d}(\log\rho V_M)^{2/d},
\end{equation}
and at $h$-distance $\geq c_*\lambda$ from $\partial M$ (or $\partial M=\emptyset$), the count satisfies
\begin{equation}\label{eq:F2-bound}
   \bigl||f(C)\cap D|-\rho\,\mathrm{Vol}_g(D)\bigr|
   \;\leq\;\delta_D\cdot\rho\,\mathrm{Vol}_g(D),
\end{equation}
where the tolerance scales with diamond volume as
\begin{equation}\label{eq:F2-tolerance}
   \delta_D := K_d\,\sqrt{\frac{\log(\rho V_M)}{\rho\,\mathrm{Vol}_g(D)}}.
\end{equation}

\item \textbf{Distance preservation:}\footnote{We retain (F3) as a separate hypothesis, though its logical relationship to (F1)--(F2) is open. We have been unable to construct a configuration satisfying (F1)--(F2) that violates (F3), which suggests that the strong volume control of (F2) may already constrain the longest-chain distances; we expect, however, that any such implication is non-trivial, and we do not assume it. Whether (F2) implies (F3) is left to future work. All results below require only that (F3) hold, which a Poisson sprinkling guarantees almost surely (Theorem~\ref{thm:poisson-wc}).}
The discrete longest chains approximate the continuum proper time: for every $x\preceq y$ in $C$ whose images lie in a common mesoscopic diamond of proper-time height in the admissible range $[\tau_{\min},c_*\lambda]$,
\begin{equation}\label{eq:F3-bound}
   \biggl|\frac{\ell_C(x,y)}{(m_d\rho)^{1/d}}-\tau_g\bigl(f(x),f(y)\bigr)\biggr|
   \;\leq\;
   C_d\Bigl(\frac{\tau_g^2}{\lambda^2}
   +\frac{\log^{3/2}(\rho V_M)}{(\rho V_M)^{1/(2d)}}\Bigr)\cdot\tau_g,
\end{equation}
where $m_d$ is the Myrheim--Meyer constant for the Minkowski causal order, $\tau_g$ is the proper time, and $\ell_C(x,y)$ is the length of the longest causal chain between $x$ and $y$.

\end{enumerate}
\end{definition}

\subsection{The smoothing scale and scale hierarchy}
\label{subsec:scales}

The proof juggles several scales, each playing a distinct role. We summarize them here for quick reference; the table should be consulted as a glossary while reading later proofs.

\begin{definition}[Scale glossary]\label{def:scales}
The proof uses the scales listed in Table~\ref{tab:scales}, related as $\rho^{-1/d}\ll\ell\ll\alpha\ell\ll\lambda$.

\begin{table}[t]
\centering
\caption{Scales and dimensionless parameters used throughout the proof.}
\label{tab:scales}
\begin{tabular}{@{}lll@{}}
\toprule
\textbf{Symbol} & \textbf{Meaning} & \textbf{Constraint}\\
\midrule
$\rho^{-1/d}$ & Discreteness scale (mean point spacing) & Smallest meaningful scale\\
$\ell$ & Karcher mean smoothing scale & $\rho^{-1/d}\ll\ell\ll\lambda$\\
$\alpha$ & Cutoff width parameter & $\alpha\geq 16$, dimensionless\\
$\alpha\ell$ & Cutoff support radius & Effective support of $w_k(x)$\\
$\lambda$ & Curvature scale & Equation~\eqref{eq:lambda-intrinsic}\\
$c_*$ & F2 admissible range constant & Dimensional, $c_*\in(0,1)$\\
$V_M$ & Manifold $g$-volume & Finite (or precompact)\\
$\delta_D$ & F2 tolerance at scale $D$ & $K_d\sqrt{\log\rho V_M/(\rho V_D)}$\\
$\delta_\ell$ & At the smoothing scale & Specialization of $\delta_D$\\
$\varepsilon_\tau$ & Distance preservation error at scale $\alpha\ell$ & $O\bigl((\alpha\ell)^2/\lambda^2+(\rho V)^{-1/(2d)}\log^{3/2}\bigr)$, Eq.~\eqref{eq:eps-tau-scale}\\
$\varepsilon$ & Final isometry error & Equation~\eqref{eq:eps-final}\\
\bottomrule
\end{tabular}
\end{table}

\smallskip
\noindent
\textbf{Where each scale appears:}
The smoothing scale $\ell$ controls the Karcher mean smoothing (Construction~\ref{con:phi}); the cutoff support $\alpha\ell$ localizes the Karcher minimization to a strongly-convex ball (Remark~\ref{rem:aux-curv}); the anchor scaffold in the trilateration argument (Section~\ref{sec:trilat}) uses anchor positions at time offset $8\alpha\ell$ and spatial offset $2\alpha\ell$ from the center.

In Theorem~\ref{thm:haupt}, $\ell$ is optimized to $\ell_*=\rho^{-1/(5d)}\lambda^{4/5}$, achieving $\varepsilon\to 0$ as $\rho\to\infty$.
\end{definition}

We write
\begin{equation}\label{eq:Mcirc-def}
   M^\circ:=\{x\in M:\ d_h(x,\partial M)>c_*\lambda+12\alpha\ell\}
\end{equation}
for the \emph{deep interior}: the region in which every active ball $B^h_{2\alpha\ell}(x)$ and its anchor scaffold lie at macroscopic distance from the boundary. When $\partial M=\emptyset$ we set $M^\circ:=M$.

\subsection{Distance preservation}
\label{subsec:dist-pres}

The longest chain length $\ell_C(x,y)$ between causally related elements $x\prec y$ is a combinatorial invariant of the abstract causal set $C$, with no reference to any embedding, and so is the same for both embeddings $f_1,f_2$. Combined with the distance-preservation condition~(F3) of Definition~\ref{def:faithful}, this yields an approximate preservation of proper times between $(M_1,g_1)$ and $(M_2,g_2)$ that depends only on the causal set.

\begin{corollary}[Approximate distance preservation]
\label{cor:dist-pres}
Let $f_1,f_2$ be two well-conditioned embeddings of the causal set $C$ into $(M_1,g_1)$ and $(M_2,g_2)$ respectively. For $x\prec y$ in $C$ with both $f_1(x),f_1(y)$ and $f_2(x),f_2(y)$ in mesoscopic diamonds of diameter $L$ and Minkowski volume $V\sim L^d$:
\begin{equation}\label{eq:dist-pres}
   |\tau_1(f_1(x),f_1(y))-\tau_2(f_2(x),f_2(y))|
   \;\leq\;
   \varepsilon_\tau(L)\cdot\max(\tau_1,\tau_2),
\end{equation}
where $\varepsilon_\tau(L)=O\bigl(L^2/\lambda^2 +(\rho V)^{-1/(2d)}\log^{3/2}(\rho V)\bigr) \to 0$ as $\rho V\to\infty$. Throughout the construction the corollary is applied at the mesoscopic anchor scale $L\sim\alpha\ell$ (all anchor and active-point separations are $O(\alpha\ell)$), and we abbreviate the corresponding value as
\begin{equation}\label{eq:eps-tau-scale}
   \varepsilon_\tau:=\varepsilon_\tau(\alpha\ell) =O\bigl((\alpha\ell)^2/\lambda^2 +(\rho(\alpha\ell)^d)^{-1/(2d)}\log^{3/2}(\rho(\alpha\ell)^d)\bigr).
\end{equation}
\end{corollary}

\begin{proof}
By condition~(F3) of Definition~\ref{def:faithful}, both proper times approximate the common combinatorial quantity $\ell_C(x,y)/(m_d\rho)^{1/d}$ with relative error $\varepsilon_\tau(L)$; the triangle inequality gives the claim.
\end{proof}

\section{Lorentzian trilateration}\label{sec:trilat}

\paragraph{The anchor configuration.}
Fix once and for all the following target geometric configuration of $d+1$ anchor positions in $\R^{1,d-1}$. The anchors are pushed out to mesoscopic radius $8\alpha\ell$ so that the \emph{entire} active ball $B_{2\alpha\ell}(x)$ (not merely the inner core) sits strictly between the anchor cones; this is what allows the trilateration to localize every active point using only timelike data:
\begin{equation}\label{eq:anchor-config}
   a_0^*:=-8\alpha\ell\,e_0, \quad
   a_d^*:=+8\alpha\ell\,e_0, \quad
   a_m^*:=8\alpha\ell\,e_0+2\alpha\ell\,e_m
   \ \ (m=1,\ldots,d-1),
\end{equation}
where $\{e_0,e_1,\ldots,e_{d-1}\}$ is the standard Minkowski basis ($e_0$ timelike, $e_m$ spatial). The displacement vectors from $a_0^*$ are
\begin{equation}\label{eq:anchor-displacements}
   v_d^*:=a_d^*-a_0^*=16\alpha\ell\,e_0,
   \qquad
   v_m^*:=a_m^*-a_0^*=16\alpha\ell\,e_0+2\alpha\ell\,e_m \ \ (m=1,\ldots,d-1).
\end{equation}

\begin{lemma}[Properties of the target configuration]
\label{lem:target-config}
Assume $\alpha\geq 16$. The configuration~\eqref{eq:anchor-config} satisfies:
\begin{enumerate}[(C1)]
\item \textbf{Pairwise causal relations:}
$a_0^*$ is in the strict chronological past of every other $a_m^*$; specifically, $\eta(a_m^*-a_0^*,a_m^*-a_0^*) \leq-(16\alpha\ell)^2+(2\alpha\ell)^2<0$.

\item \textbf{Active-ball coverage:}
Every $p\in\R^d$ with $|p|\leq 2\alpha\ell$ satisfies $\eta(a_m^*-p,a_m^*-p)<0$ for each $m=0,\ldots,d$ (i.e., $p$ is strictly chronologically related to every anchor).

\item \textbf{Well-conditioned basis:}
The matrix $V^*$ with rows $(v_m^*)^T$ is invertible with $\|V^*\|=\Theta(\alpha\ell)$ and $\|(V^*)^{-1}\|=\Theta(1/(\alpha\ell))$, hence condition number $\kappa(V^*)=O(1)$.
\end{enumerate}
\end{lemma}

\begin{proof}[Proof sketch]
Each of (C1)--(C3) is a direct verification using the explicit form~\eqref{eq:anchor-config} with $\alpha\geq 16$:
\begin{itemize}
\item (C1) follows from $\eta(v_d^*,v_d^*)=-(16\alpha\ell)^2$ and $\eta(v_m^*,v_m^*)=-(16\alpha\ell)^2+(2\alpha\ell)^2<0$.

\item (C2) holds because the time offset $8\alpha\ell$ of each upper anchor exceeds the spatial offset ($2\alpha\ell$) plus the active-ball radius ($2\alpha\ell$) plus the active-ball half-height ($2\alpha\ell$): the worst-case top point $p^0=+2\alpha\ell$, $|\vec p|=2\alpha\ell$ satisfies $(8-2)\alpha\ell=6\alpha\ell>(2+2)\alpha\ell =4\alpha\ell$, so $\eta(a_m^*-p,a_m^*-p)<0$; the lower anchor at $-8\alpha\ell$ is covered with even larger margin.

\item (C3) follows by writing the basis matrix $V^*$ explicitly. Because the timelike and spacelike components of the anchor coordinates scale strictly proportionally, the geometric scale factor $\alpha\ell$ factors out of the matrix entirely. Since condition numbers are invariant under scalar multiplication, the condition number of $V^*$ is determined solely by a fixed matrix of dimensionless constants, making it dimension-dependent but strictly independent of $\alpha$.
\end{itemize}
Full numerical inequalities in Appendix~\ref{app:anchor-checks}.
\end{proof}

\paragraph{Lorentzian trilateration identity.}

\begin{lemma}[Trilateration identity]
\label{lem:trilat-identity}
Let $a_0,\ldots,a_d\in\R^d$ with $a_0=0$, set $v_m:=a_m$ for $m=1,\ldots,d$, and assume $\{v_m\}$ is a basis of $\R^d$. For any $p\in\R^d$, define $D_m(p):=\eta(p-a_m,p-a_m)$. Then the Lorentzian inner products $\eta(p,v_m)$ satisfy
\begin{equation}\label{eq:trilat-id}
   2\,\eta(p,v_m)
   = \eta(v_m,v_m) - \bigl(D_m(p)-D_0(p)\bigr),
   \qquad m=1,\ldots,d.
\end{equation}
Equivalently, $p=(V\eta)^{-1}\mathbf{w}(p)$ where $V$ has rows $v_m^T$ and $\mathbf{w}(p)_m =\tfrac{1}{2}\bigl[\eta(v_m,v_m)-D_m(p)+D_0(p)\bigr]$.
\end{lemma}
\begin{proof}
Expanding and using polarization $D_m(p)-D_0(p) =\eta(p-a_m,p-a_m)-\eta(p,p)=-2\eta(p,v_m)+\eta(v_m,v_m)$ gives~\eqref{eq:trilat-id}. The map $p\mapsto(\eta(p,v_m))_m =V\eta p$ is invertible since $V$ is invertible (basis hypothesis) and $\eta$ is non-degenerate.
\end{proof}

\paragraph{Approximate rigidity from anchor distances.} We first isolate the key analytic ingredient: a quantitative rigidity bound (Lemma~\ref{lem:approx-az}) showing that any matrix which approximately preserves the Minkowski form is close to an exact Lorentz transformation. It is purely a stability statement for the quadratic-form map $B\mapsto B^T\eta B$ near $O(1,d-1)$, depending only on linear algebra and the implicit function theorem (proof in Appendix~\ref{app:approx-az}), independent of the trilateration argument. We then use it to establish the quantitative finite-point analogue of the Alexandrov--Zeeman theorem (Lemma~\ref{lem:finite-procrustes}): approximate preservation of Lorentzian distances on a well-conditioned scaffold forces a Lorentz transformation up to controlled error.

\begin{lemma}[Approximate Alexandrov--Zeeman]
\label{lem:approx-az}
Let $A\colon\R^{1,d-1}\to\R^{1,d-1}$ be an invertible
linear map that is uniformly conditioned,
$\kappa(A):=\|A\|\,\|A^{-1}\|\leq K$. If $A$ maps the null cone
$\mathcal{N}=\{v:\eta(v,v)=0\}$ to within angular error
$\theta\leq\theta_0$ (i.e., for every null $v$,
$|\eta(Av,Av)|/|Av|^2\leq\theta$), then there exist
$\Omega>0$ and $\Lambda\in O(1,d-1)$ with
\[
   \|A-\Omega\Lambda\|\leq C\,\|A\|\,\theta,
\]
where the threshold $\theta_0>0$ and the constant $C>0$ depend
only on the dimension~$d$ and the conditioning bound~$K$.
\end{lemma}
The proof, a stability statement for the quadratic-form map $B\mapsto B^T\eta B$ near
$O(1,d-1)$, is deferred to Appendix~\ref{app:approx-az}. The
uniform-conditioning hypothesis is essential rather than cosmetic:
since $O(1,d-1)$ is noncompact, the stability constant degrades for
ill-conditioned $A$ (large Lorentz boosts), so a bound
$\kappa(A)\leq K$ cannot be dispensed with. At each call site $A$ is
shown to satisfy $\kappa(A)=O(1)$, so $C$ and $\theta_0$ are
effectively dimensional there.

\begin{lemma}[Finite Lorentzian Procrustes\footnote{In matrix approximation, the Orthogonal Procrustes problem is to find an orthogonal transformation that most closely maps a given matrix to a target. The term derives from the bandit of Attica who forced his guests to fit his bed by stretching or cutting them; for the classical account, see Plutarch, \textit{Life of Theseus}, 11.1, Pausanias, \textit{Description of Greece}, 38.5.}]
\label{lem:finite-procrustes}
Let $\{a_0,\ldots,a_d\}\subset\R^d$ satisfy properties (C1)--(C3) of Lemma~\ref{lem:target-config} (we may take $a_0=0$ by translation). Here and below $\|\cdot\|$ denotes the Euclidean norm on $\R^d$ together with the operator norm it induces on $d\times d$ matrices. Let $\{b_0,\ldots,b_d\}\subset\R^d$ with $b_0=0$ satisfy the boundedness hypothesis
\begin{equation}\label{eq:proc-bdd}
   \|b_m\|\leq C_{\mathrm{anch}}\,\alpha\ell
   \quad\text{for each }m=0,\ldots,d,
\end{equation}
where $C_{\mathrm{anch}}$ is a fixed dimensional constant (in our application, $C_{\mathrm{anch}}=12$, coming from Lemma~\ref{lem:scaffold}, which places each $b_m$ within $\varepsilon_\tau\alpha\ell\ll\alpha\ell$ of an ideal target point $b_m^*$ satisfying $\|b_m^*\|\leq 8\sqrt2\,\alpha\ell$; the triangle inequality then gives $\|b_m\|\leq\|b_m^*\|+\varepsilon_\tau\alpha\ell\leq(8\sqrt2+\varepsilon_\tau)\alpha\ell\leq 12\alpha\ell$ for $\rho$ large). Suppose further
\begin{equation}\label{eq:proc-hyp}
   \bigl|\eta(b_i-b_j,b_i-b_j)-\eta(a_i-a_j,a_i-a_j)\bigr|
   \leq\varepsilon
   \quad\text{for all }0\leq i,j\leq d,
\end{equation}
with $\varepsilon\leq c_0(\alpha\ell)^2$ for $c_0$ sufficiently small. Then there exists $\hat\Lambda\in O(1,d-1)$ such that
\begin{equation}\label{eq:proc-conc}
   \|b_m-\hat\Lambda a_m\|\leq C_d\,\frac{\varepsilon}{\alpha\ell}
   \quad\text{for each }m=0,\ldots,d,
\end{equation}
where $C_d$ is a dimensional constant. With $\varepsilon=\varepsilon_\tau(\alpha\ell)^2$ this reads $\|b_m-\hat\Lambda a_m\|\leq C_d\,\varepsilon_\tau\,\alpha\ell$.

Moreover, the projected transformation is uniformly conditioned:
\begin{equation}\label{eq:proc-boost}
   \|\hat\Lambda\|\leq C_\Lambda,
   \qquad
   \sigma_{\min}(\hat\Lambda)\geq c_\Lambda,
\end{equation}
for dimensional constants $C_\Lambda\geq c_\Lambda>0$ depending only on $d$ (and on the fixed ratios of the ideal configuration), not on $\alpha$, $\ell$, or the basepoint.
\end{lemma}
\begin{proof}[Proof sketch]
Form the matrices $A:=[a_1\cdots a_d]$ and $B:=[b_1\cdots b_d]$ of column anchors. By Lemma~\ref{lem:target-config}(C3), $A$ is invertible with $\|A^{-1}\|=O(1/(\alpha\ell))$. The candidate $\Lambda:=BA^{-1}$ satisfies $b_m=\Lambda a_m$ exactly, but is not literally Lorentz. By polarization, the hypothesis~\eqref{eq:proc-hyp} gives $\|A^T\eta A-B^T\eta B\|=O(\varepsilon)$, which transfers to $\|\Lambda^T\eta\Lambda-\eta\|=O(\varepsilon/(\alpha\ell)^2)$ via the bound on $\|A^{-1}\|$. Since $\Lambda=BA^{-1}$ is uniformly conditioned ($\kappa(\Lambda)=O(1)$, from the anchor size bounds in the detailed proof below), Lemma~\ref{lem:approx-az} then projects $\Lambda$ onto $O(1,d-1)$ with $\|\Lambda-\hat\Lambda\|=O(\varepsilon/(\alpha\ell)^2)$, and the pointwise residual $\|b_m-\hat\Lambda a_m\| \leq\|\Lambda-\hat\Lambda\|\cdot\|a_m\| =O(\varepsilon/(\alpha\ell))$ follows. With $\varepsilon=\varepsilon_\tau(\alpha\ell)^2$ this gives the stated bound.

\emph{Uniform conditioning of $\hat\Lambda$.} Although $\hat\Lambda\in O(1,d-1)$ is invertible by construction, the group $O(1,d-1)$ is noncompact, so invertibility alone does not bound its singular values; a pure boost of rapidity $\xi$ has singular values $e^{\pm\xi}$. We bound the boost using the boundedness hypothesis~\eqref{eq:proc-bdd}. From $b_m=\hat\Lambda a_m+e_m$, in which $e_m$ denotes the residual error vector (not a coordinate basis vector) with $\|e_m\|\leq C_d\varepsilon_\tau\alpha\ell$, and $\|b_m\|\leq C_{\mathrm{anch}}\alpha\ell$, we get $\|\hat\Lambda a_m\|\leq(C_{\mathrm{anch}}+C_d\varepsilon_\tau)\alpha\ell \leq 2C_{\mathrm{anch}}\alpha\ell$ for $\rho$ large. Since the $\{a_m\}$ contain a basis matrix $V$ with $\sigma_{\min}(V)\geq c_d\alpha\ell$ (Lemma~\ref{lem:target-config}(C3)), every unit vector $u$ is a combination $u=\sum_m \beta_m (a_m/\|a_m\|)$ with $|\beta|\leq\|V^{-1}\|\cdot\|u\|\cdot\max_m\|a_m\| =O(1)$; hence $\|\hat\Lambda u\|\leq\sum_m|\beta_m|\,\|\hat\Lambda a_m\|/\|a_m\| \leq C_\Lambda$, giving $\|\hat\Lambda\|=\sigma_{\max}(\hat\Lambda) \leq C_\Lambda$. Finally, since $\hat\Lambda\in O(1,d-1)$ has $|\det\hat\Lambda|=1$ and singular values $\sigma_1\geq\cdots\geq\sigma_d>0$ with $\prod_i\sigma_i=1$,
\[
   \sigma_{\min}(\hat\Lambda)=\sigma_d
   =\Bigl(\prod_{i<d}\sigma_i\Bigr)^{-1}
   \geq\sigma_{\max}(\hat\Lambda)^{-(d-1)}
   \geq C_\Lambda^{-(d-1)}=:c_\Lambda>0,
\]
which is~\eqref{eq:proc-boost}. Geometrically, the target anchors are confined to a ball of radius $C_{\mathrm{anch}}\alpha\ell$, so $\hat\Lambda$ cannot stretch the order-$\alpha\ell$ source frame by more than an $O(1)$ factor, capping the boost rapidity at $O(1)$. Full algebraic details in Appendix~\ref{app:procrustes-algebra}.
\end{proof}

\begin{remark}[Preservation of anchor separations]
\label{rem:anchor-pres}
The hypothesis~\eqref{eq:proc-hyp} holds for the embedded anchors with $\varepsilon=O(\varepsilon_\tau)(\alpha\ell)^2$, and this is where the scaffold conditioning (Lemma~\ref{lem:scaffold}) does work that (F3) cannot. The causal set records proper times only along chains, so (F3) constrains the \emph{timelike} anchor pairs (those involving $a_0$), preserving $\eta(a_0-a_m,a_0-a_m)=-\tau_g(a_0,a_m)^2$ to relative error $\varepsilon_\tau$ by Corollary~\ref{cor:dist-pres}. The configuration also contains \emph{spacelike} pairs (e.g.\ $a_i,a_j$ with $i,j\in\{1,\ldots,d-1\}$), about which (F3) is silent. These are controlled by Lemma~\ref{lem:scaffold}: the source anchors $\{a_m\}$ and the target anchors $\{b_m\}$ each lie within $\varepsilon_\tau\alpha\ell$ of the \emph{same} ideal configuration~\eqref{eq:anchor-config}, so for every pair $i,j$, timelike or spacelike,
\[
   \bigl|\eta(b_i-b_j,b_i-b_j)-\eta(a_i-a_j,a_i-a_j)\bigr|
   \leq C_d\,(\alpha\ell)\cdot\varepsilon_\tau\alpha\ell
   =O(\varepsilon_\tau)\,(\alpha\ell)^2 .
\]
This is precisely why the scaffold proximity radius in Lemma~\ref{lem:scaffold} is $O(\varepsilon_\tau\alpha\ell)$: a coarser radius $\sim\ell$ would bound the spacelike separations only to $O(\alpha\ell^2)$, exceeding the trilateration tolerance by the diverging factor $1/(\varepsilon_\tau\alpha)$ and stalling the final rate at a constant floor rather than letting it tend to zero.
\end{remark}

\paragraph{Extending to all causally-related points.}

\begin{lemma}[Trilateration extension]
\label{lem:trilat-extension}
Let $\{a_m\},\{b_m\},\hat\Lambda$ be as in Lemma~\ref{lem:finite-procrustes}, with $\|b_m-\hat\Lambda a_m\|\leq E_a$ for each $m$. Suppose $p\in\R^d$ is chronologically related to every $a_m$, and
$q\in\R^d$ satisfies
\begin{equation}\label{eq:trilat-hyp}
   \bigl|\eta(q-b_m,q-b_m)-\eta(p-a_m,p-a_m)\bigr|\leq E_p
   \quad\text{for }m=0,\ldots,d.
\end{equation}
Then
\begin{equation}\label{eq:trilat-conc}
   \|q-\hat\Lambda p\|
   \leq C_d\,\bigl(E_p/(\alpha\ell)+E_a\bigr).
\end{equation}
\end{lemma}
\begin{proof}
By Lemma~\ref{lem:trilat-identity}, with $a_0=0$ as our choice of origin, $p$ is determined by the $\eta$-inner products $\eta(p,v_m)$ via $V\eta p=\mathbf{w}(p)$ where $\mathbf{w}(p)_m=\tfrac{1}{2}[\eta(v_m,v_m)-D_m(p)+D_0(p)]$.

Apply the same identity in the target space: with $b_0=0$ and $v_m'=b_m$, $q=(V'\eta)^{-1}\mathbf{w}'(q)$ where $V'$ has rows $(v_m')^T$ and analogously for $\mathbf{w}'(q)$.

Compute $\hat\Lambda p=\hat\Lambda(V\eta)^{-1}\mathbf{w}(p) =(V\eta\hat\Lambda^{-1})^{-1}\mathbf{w}(p)$. Since $\hat\Lambda\in O(1,d-1)$, $\hat\Lambda^T\eta\hat\Lambda =\eta$, equivalently $\eta\hat\Lambda^{-1}=\hat\Lambda^T\eta$. Hence $V\eta\hat\Lambda^{-1}=V\hat\Lambda^T\eta$, the $m$-th row of which is $v_m^T\hat\Lambda^T\eta=(\hat\Lambda v_m)^T\eta$.

By hypothesis, $b_m=\hat\Lambda a_m+e_m$ with $\|e_m\|\leq E_a$. Hence $v_m'=b_m=\hat\Lambda a_m+e_m=\hat\Lambda v_m+e_m$ (using $a_0=b_0=0$ so $v_m=a_m$, $v_m'=b_m$). The $m$-th row of $V'$ is $(v_m')^T=(\hat\Lambda v_m)^T+e_m^T$, so $V'-V\hat\Lambda^T$ has rows $e_m^T$, giving $\|V'-V\hat\Lambda^T\|\leq\sqrt{d}\cdot E_a$. Right-multiplying by $\eta$ (norm $1$): $\|V'\eta-V\hat\Lambda^T\eta\| =\|V'\eta-V\eta\hat\Lambda^{-1}\|\leq\sqrt{d}E_a$.

For $\mathbf{w}(p)-\mathbf{w}'(q)$: each component is $\tfrac{1}{2}[\eta(v_m,v_m)-\eta(v_m',v_m')] -\tfrac{1}{2}[D_m(p)-D_m'(q)]+\tfrac{1}{2}[D_0(p)-D_0'(q)]$.
The first bracket is bounded by direct estimation:
\[
   |\eta(v_m,v_m)-\eta(v_m',v_m')|
   \leq\|\eta\|\cdot(\|v_m\|+\|v_m'\|)\cdot\|v_m-v_m'\| \leq 2\|v_m\|\cdot E_a+O(E_a^2) = O(\alpha\ell\cdot E_a)
\]
for $E_a\ll\alpha\ell$. The remaining two brackets are bounded by~\eqref{eq:trilat-hyp}: each is $\leq E_p$. Total: $\|\mathbf{w}(p)-\mathbf{w}'(q)\|\leq O(\alpha\ell\cdot E_a+E_p)$.

Now
\begin{align*}
   q-\hat\Lambda p
   &=(V'\eta)^{-1}\mathbf{w}'(q)-(V\eta\hat\Lambda^{-1})^{-1}\mathbf{w}(p)\\
   &=(V'\eta)^{-1}[\mathbf{w}'(q)-\mathbf{w}(p)]
   +[(V'\eta)^{-1}-(V\eta\hat\Lambda^{-1})^{-1}]\mathbf{w}(p).
\end{align*}
We bound each term.

\emph{First term:} $\|(V'\eta)^{-1}\|=O(1/(\alpha\ell))$ (since (C3) holds for $V'$ up to $O(E_a/(\alpha\ell))$ perturbation, still $\Theta(\alpha\ell)$ in the relevant norm). Times $\|\mathbf{w}'(q)-\mathbf{w}(p)\| =O(\alpha\ell\cdot E_a + E_p)$ gives $O(E_a+E_p/(\alpha\ell))$.

\emph{Second term:} Using
$X^{-1}-Y^{-1}=X^{-1}(Y-X)Y^{-1}$:
$\|(V'\eta)^{-1}-(V\eta\hat\Lambda^{-1})^{-1}\| \leq\|(V'\eta)^{-1}\|\cdot\|V\eta\hat\Lambda^{-1}-V'\eta\| \cdot\|(V\eta\hat\Lambda^{-1})^{-1}\|$ $=O(1/(\alpha\ell))\cdot O(E_a)\cdot O(1/(\alpha\ell)) =O(E_a/(\alpha\ell)^2)$. Times $\|\mathbf{w}(p)\| =O((\alpha\ell)^2)$ gives $O(E_a)$.

Summing: $\|q-\hat\Lambda p\|\leq O(E_a+E_p/(\alpha\ell))$.
\end{proof}

\paragraph{Target localization.}

The trilateration extension yields, as an immediate corollary, the localization of the active target points. The Karcher mean defining $\Phi(x)$ is only uniquely well-defined if the target points are tightly confined. Because this localization relies strictly on timelike proper times to the anchor scaffold, it can be carried out entirely prior to the construction of the map $\Phi$.

\begin{corollary}[Active target points are localized]
\label{cor:target-localization}
Fix $x\in M_1^\circ$ and work in $h_1$-normal coordinates at $x$. Under Lemma~\ref{lem:scaffold} (the anchors $\{a_m\}$, themselves embedded points of $f_1(C)$, exist with (C1)--(C4)) and with $\hat\Lambda\in O(1,d-1)$ the Lorentz transformation approximating the coordinate map $p_k\mapsto q_k$ (Lemma~\ref{lem:finite-procrustes}), every active target point $q_k=(f_2\circ f_1^{-1})(p_k)\in M_2$ (the image of the embedded source point $p_k\in M_1$) with $w_k(x)>0$ satisfies
\begin{equation}\label{eq:target-loc}
   \bigl\|q_k-\hat\Lambda p_k\bigr\| \leq C_d\,\varepsilon_\tau\,\alpha\ell .
\end{equation}
Consequently all active $q_k$ lie within $h_2$-distance $R_*:=\|\hat\Lambda\|\cdot 2\alpha\ell+C_d\varepsilon_\tau\alpha\ell \leq C_d'\,\alpha\ell\ll\mathrm{inj}(M_2,h_2)$ of the common point $q_0:=\hat\Lambda\cdot 0$ (the image of the anchor origin), for a dimensional constant $C_d'$.
\end{corollary}
\begin{proof}
Each active source point satisfies $|p_k|\leq 2\alpha\ell$. By condition (C4), these points are strictly chronologically related to every anchor $a_m$, ensuring the proper times $\tau_{g_1}(p_k,a_m)$ are well-defined. Moreover, each anchor sits at time offset $8\alpha\ell$ and spatial offset at most $2\alpha\ell$ from $x$ (the ideal configuration~\eqref{eq:anchor-config}, up to the $O(\varepsilon_\tau\alpha\ell)$ scaffold proximity of Lemma~\ref{lem:scaffold}), while $|p_k|\leq 2\alpha\ell$. Hence the coordinate separation $p_k-a_m$ has time component at most $2\alpha\ell+8\alpha\ell=10\alpha\ell$ and spatial component at most $2\alpha\ell+2\alpha\ell=4\alpha\ell$, so $|p_k-a_m|\leq\sqrt{10^2+4^2}\,\alpha\ell=\sqrt{116}\,\alpha\ell\leq 11\alpha\ell$ for $\rho$ large. Consequently the trilateration data $\mathbf{w}(p_k)$ of Lemma~\ref{lem:trilat-identity} obeys $\|\mathbf{w}(p_k)\|=O((\alpha\ell)^2)$ uniformly over the entire active ball (not merely the inner core); this is the conditioning input that makes the extension bound below apply to all active points. By Corollary~\ref{cor:dist-pres}, their squares (the quantities $-\eta(p_k-a_m,p_k-a_m)$) are preserved up to relative error $\varepsilon_\tau$, giving the hypothesis~\eqref{eq:trilat-hyp} of Lemma~\ref{lem:trilat-extension} with $E_p=O(\varepsilon_\tau(\alpha\ell)^2)$. With $E_a=O(\varepsilon_\tau\alpha\ell)$ from Lemma~\ref{lem:finite-procrustes} (whose hypothesis~\eqref{eq:proc-hyp} is verified for the anchors in Remark~\ref{rem:anchor-pres}), the conclusion~\eqref{eq:trilat-conc} gives $\|q_k-\hat\Lambda p_k\|\leq C_d(E_p/(\alpha\ell)+E_a) = O(\varepsilon_\tau\alpha\ell)$, which is~\eqref{eq:target-loc}. The spread bound follows from $\|\hat\Lambda p_k\|\leq\|\hat\Lambda\|\,|p_k| \leq\|\hat\Lambda\|\cdot 2\alpha\ell$ and the bound $\|\hat\Lambda\|=O(1)$ established in Lemma~\ref{lem:finite-procrustes}.
\end{proof}

\section{Part I: Well-Conditioned $\Rightarrow$ Uniqueness}

\paragraph{Roadmap} Throughout this section we assume $C$ admits a well-conditioned embedding $f_i$ into each of $(M_1,g_1)$ and $(M_2, g_2)$ at common density $\rho$; all statements are deterministic consequences. The argument proceeds in three stages. \textbf{Stage A} (local): using the trilateration of \S 3 together with the covariance non-degeneracy forced by F2, we construct a smooth map $\Phi: M_1^\circ \to M_2^\circ$ as a moving Riemannian center of mass, and show it is a local diffeomorphism. \textbf{Stage B} (global): a degree argument upgrades the local diffeomorphism to a global one. \textbf{Stage C} (rigidity): preservation of the longest-chain proper times forces the local Procrustes map into $O(1,d-1)$, so $\Phi$ is directly an approximate isometry; the proper times carry the scale, so no separate volume-pinning step is needed. The hypotheses entering each stage are local and finite; the conclusion, a global approximate isometry of the metric tensors, is what the stages assemble. 

\begin{lemma}[Anchor scaffolds from volume-faithfulness]
\label{lem:scaffold}
Let $f\colon C\to(M,g)$ be a faithful embedding (conditions (F1)--(F2)) satisfying the
scale-separation hypothesis $\rho\lambda^d\geq c_*^{-2d}(\log\rho V_M)^2$. Then there is a dimensional constant $C_d$ such that for every $x\in M^\circ$ there exist $d+1$ embedded points $a_0,\ldots,a_d\in f(C)$ with
\begin{equation}\label{eq:scaffold-prox}
   d_h\bigl(a_m,\,x+a_m^*\bigr)\leq \varepsilon_\tau\,\alpha\ell
   \qquad(m=0,\ldots,d),
\end{equation}
where $\{a_m^*\}$ is the ideal configuration~\eqref{eq:anchor-config}, such that the frame $V$ with rows $(a_m-a_0)^T$ satisfies
\begin{equation}\label{eq:scaffold-cond}
   \sigma_{\min}(V)\geq c_d\,\alpha\ell,\qquad \kappa(V)\leq C_d,
\end{equation}
uniformly in $x$. Moreover each $a_m$ is strictly timelike-related in $g$ to the entire active ball $B^h_{2\alpha\ell}(x)$ (property (C4)).
\end{lemma}

\begin{remark}
The scaffold is not an independent hypothesis: it is forced by the volume condition. The
mechanism is a separation of scales. The proximity radius $\varepsilon_\tau\alpha\ell$ is coarser than the resolution $\tau_{\min}$ at which (F2) operates, so (F2) compels an embedded point to lie near each ideal anchor; yet it is finer than the frame scale $\alpha\ell$, so wherever those points fall the frame remains well-conditioned.
\end{remark}

\begin{proof}
\emph{Existence (the radius exceeds $\tau_{\min}$).} Fix $x\in M^\circ$ and an ideal anchor $a_m^*$. Consider the causal diamond $D_m$ of proper-time height $\varepsilon_\tau\alpha\ell$ centered at $x+a_m^*$. The ideal anchors satisfy $\|a_m^*\|\leq 8\sqrt2\,\alpha\ell<12\alpha\ell$, so $x+a_m^*$, and with it all of $D_m$ (whose own extent $\varepsilon_\tau\alpha\ell$ is negligible), lies within $h$-distance $12\alpha\ell$ of $x$. Since $x\in M^\circ$ means $d_h(x,\partial M)>c_*\lambda+12\alpha\ell$, this gives $d_h(D_m,\partial M)>c_*\lambda$, so (F2) applies on $D_m$ and $D_m$ is an admissible diamond; the same bound keeps every scaffold diamond, each centered within $12\alpha\ell$ of $x$, interior to $M$. The proper-time height satisfies $\tau_{\min}\leq\varepsilon_\tau\alpha\ell\leq\alpha\ell$, so the scaffold scale $\alpha\ell$ itself exceeds $\tau_{\min}$; the substantive inequality $\varepsilon_\tau\alpha\ell\geq\tau_{\min}$ holds because
\begin{equation}\label{eq:scale-window}
   \frac{\varepsilon_\tau\alpha\ell}{\tau_{\min}}
   =c_*\sqrt{\alpha\,\ell}\;\rho^{1/(2d)}\,(\log\rho V_M)^{3/2-2/d}
   \;\longrightarrow\;\infty,
\end{equation}
because the factor $\rho^{1/(2d)}$ is a positive power of $\rho$ and dominates every power of $\log\rho V_M$ (and, at the optimal smoothing scale $\ell\sim\rho^{-1/(5d)}$, the residual $\sqrt{\ell}\,\rho^{1/(2d)}\sim\rho^{2/(5d)}\to\infty$). Hence $D_m$ is F2-admissible, and by
(F2) it contains
\[
   |f(C)\cap D_m|=\rho\,\Vol_g(D_m)\bigl(1+O(\delta_{D_m})\bigr)
   =\Theta\!\bigl(\rho\,(\varepsilon_\tau\alpha\ell)^d\bigr)\gg 1
\]
embedded points. Choosing any one, call it $a_m$, gives a point with $d_h(a_m,x+a_m^*)\leq\varepsilon_\tau\alpha\ell$, which is~\eqref{eq:scaffold-prox}.

\emph{Conditioning (the radius is small relative to $\alpha\ell$).} Let $V^*$ be the ideal frame with rows $(a_m^*-a_0^*)^T$. By the explicit configuration \eqref{eq:anchor-config} (Appendix~\ref{app:anchor-checks}), $V^*$ has $\sigma_{\min}(V^*)\geq c_d^*\,\alpha\ell$ and $\kappa(V^*)\leq C_d^*$ for dimensional constants. The selected frame $V$ differs from $V^*$ row-by-row by at most $2\varepsilon_\tau\alpha\ell$ (two anchor perturbations of size $\varepsilon_\tau\alpha\ell$ each), so $\|V-V^*\|\leq 2\sqrt{d}\,\varepsilon_\tau\alpha\ell$. By Weyl's inequality,
\[
   \sigma_{\min}(V)\geq\sigma_{\min}(V^*)-\|V-V^*\|
   \geq c_d^*\alpha\ell-2\sqrt d\,\varepsilon_\tau\alpha\ell
   \geq\tfrac12 c_d^*\,\alpha\ell
\]
for $\rho\lambda^d$ large (so that $\varepsilon_\tau\leq c_d^*/(4\sqrt d)$), and $\sigma_{\max}(V)\leq\sigma_{\max}(V^*)+\|V-V^*\|\leq 2\sigma_{\max}(V^*)$, whence $\kappa(V)\leq 4\kappa(V^*)\leq C_d$. This is~\eqref{eq:scaffold-cond} with $c_d=\tfrac12 c_d^*$; the constants $c_d,C_d$ depend only on the dimension $d$, through the fixed dimensionless configuration~\eqref{eq:anchor-config}, and not on $\rho$, $\lambda$, or $\ell$.

\emph{Chronology (C4).} Each $a_m$ lies within $\varepsilon_\tau\alpha\ell$ of $x+a_m^*$, and the ideal configuration places the anchors at time offset $8\alpha\ell$ from $x$, exceeding the spatial reach $2\alpha\ell$ plus the active-ball half-height $2\alpha\ell$ with margin. Since $\alpha\ell\ll\lambda$, the curvature distortion $g=\eta+O(\alpha^2\ell^2/\lambda^2)$ and the $\varepsilon_\tau\alpha\ell$ perturbation are too small to flip any strict causal relation, so every $a_m$ is strictly timelike-related to all of $B^h_{2\alpha\ell}(x)$. Uniformity over $x\in M^\circ$ follows since all constants are dimensional and (F2) holds uniformly over admissible diamonds.
\end{proof}

\begin{lemma}[Covariance isotropy from volume-faithfulness]
\label{lem:source-cov}
Let $f\colon C\to(M,g)$ be a well-conditioned embedding (Definition~\ref{def:faithful}), and fix $x\in M^\circ$. Work in $h$-normal coordinates at $x$, and write $\tilde w_k:=\exp_x^{-1}(p_k)$ for the active embedded points and, with $d_k:=d_h(x,p_k)$, $w_k(x):=\chi(d_k/\alpha\ell)\exp(-d_k^2/\ell^2)$ for the cutoff-Gaussian weights. Let
\begin{equation}\label{eq:Sigma-def}
   \Sigma(x):=\frac{1}{W}\sum_k w_k(x)\,\tilde w_k\tilde w_k^T, \qquad W:=\sum_k w_k(x),
\end{equation}
be the source covariance. Then there is a dimensional constant $c_d>0$ such that, for $\rho\lambda^d$ sufficiently large, $\Sigma(x)$ is isotropic up to lower-order error,
\begin{equation}\label{eq:sigma-iso}
   \bigl\|\Sigma(x)-c_d\,\ell^2\,I\bigr\|
   \;\leq\; C_d\,\ell^2\,\delta_{\mathrm{cov}},
\end{equation}
and in particular non-degenerate,
\begin{equation}\label{eq:sigma-bound}
   \sigma_{\min}\bigl(\Sigma(x)\bigr)
   \;\geq\; c_d\,\ell^2\bigl(1-C_d\,\delta_{\mathrm{cov}}\bigr)
   \;\geq\;\tfrac12 c_d\,\ell^2,
\end{equation}
both uniformly over $x\in M^\circ$. Here
\begin{equation}\label{eq:delta-cov}
   \delta_{\mathrm{cov}}:=\Bigl(\frac{\log\rho V_M}{\rho\,\ell^d}\Bigr)^{1/(d+2)}\longrightarrow 0
\end{equation}
is the error contributed by the F2 volume-count tolerance, minimized over the scale of the dense cover used in the proof. It is polynomially small in $\rho$ and depends on $\ell$, in contrast to the tolerance $\delta_{\tau_{\min}}=K_d c_*^{d/2}/\sqrt{\log\rho V_M}$ of a single finest-scale diamond, which is only logarithmically small and independent of $\ell$; pinning the cover at the finest scale $\tau_{\min}$ would replace $\delta_{\mathrm{cov}}$ by the larger $\delta_{\tau_{\min}}+\tau_{\min}/\ell$.
\end{lemma}

\begin{remark}
The point of Lemma~\ref{lem:source-cov} is that isotropy (and hence non-degeneracy) of the local point cloud is forced by the volume condition (F2) alone, with no appeal to the random nature of the embedding or to the map $\Phi$ constructed later. Geometrically, the lower bound is supplied by the positions of a grid of mesoscopic diamonds filling the inner core; the F2 count tolerance together with the freedom remaining within each diamond perturb the bound at relative order $\delta_{\mathrm{cov}}$ once the cover scale is optimized (\eqref{eq:delta-cov}), which vanishes polynomially in $\rho$, rather than under the weaker requirement that admissible diamonds merely exist ($\tau_{\min}\leq c_*\lambda$).
\end{remark}

\begin{proof}
\emph{Step~1: Setup.}
The support of the weights is the ball $B_{2\alpha\ell}^h(x)$, which is a mesoscopic region with $2\alpha\ell\leq c_*\lambda$ under scale separation; F2 applied to a diamond covering this ball gives the upper bound $W\leq C_d\,\rho\ell^d$.

\emph{Step~2: A dense admissible cover at a free scale.}
Fix a cover scale $s\in[\tau_{\min},\ell]$, to be optimized in Step~4; each diamond below has proper-time height $s$ and is therefore F2-admissible (since $s\geq\tau_{\min}$). Let $B:=\{|\tilde w|\leq\ell\}$ be the inner core, on which $w_k\geq e^{-1}$. In $h$-normal coordinates the flat Alexandrov diamond of proper-time height $s$ centered at $c$,
\[
   D(c):=\bigl\{p:\ |p^0-c^0|+|\vec p-\vec c|\leq s/2\bigr\},
\]
is the Minkowski causal diamond (Alexandrov interval) of proper-time height $s$: a bicone with tips $c\mp\tfrac{s}2 e_0$, equal to the ball of radius $s/2$ for the norm $N(p):=|p^0|+|\vec p|$, of Euclidean volume $\kappa_d s^d$. The curved causal diamond differs from it by a relative $O(s^2/\lambda^2)$ cone-aperture distortion (Step~1 of Proposition~\ref{prop:longest-chain}), absorbed into the error below. Each $D(c)$ has count $N(c)=\rho\kappa_d s^d(1\pm\delta_s)$, where $\delta_s:=K_d\sqrt{\log\rho V_M/(\rho\kappa_d s^d)}$ is the F2 tolerance at scale $s$; the $(1\pm\delta)$ notation abbreviates the two-sided F2 bound, which constrains the integer count without pinning it.

The essential point is that F2 constrains \emph{every} admissible diamond, including those centered in the interstices of any chosen family, so we use an overlapping cover. Fix $\varepsilon\in(0,1)$ and place centers on the fine grid $\{c_j\}:=\varepsilon s\,\Z^d$, retaining those with $D(c_j)\cap B\neq\emptyset$. Every $p\in B$ then lies in
\[
   M(p):=\#\{j:\ p\in D_j\}
   =\frac{\kappa_d s^d}{(\varepsilon s)^d}\bigl(1+O(\varepsilon)\bigr)
   =\kappa_d\,\varepsilon^{-d}\bigl(1+O(\varepsilon)\bigr)
\]
of the diamonds (grid points in an $N$-ball of radius $s/2$, with relative lattice fluctuation $O(\varepsilon)$), uniformly for $p$ in the core interior; within the shell of width $s$ about $\partial B$ the multiplicity dips, contributing a relative $O(s/\ell)$ absorbed below.

\emph{Step~3: Two-sided isotropy~\eqref{eq:sigma-iso}.} We carry the construction at the level of the full matrix rather than a single direction. Define the continuum reference second moment
\[
   \bar\Sigma:=\frac{\int_{\R^d} p\,p^T\,e^{-|p|^2/\ell^2}\,dp}
                    {\int_{\R^d} e^{-|p|^2/\ell^2}\,dp}
            =\tfrac12\ell^2 I,
\]
which is isotropic by rotational symmetry of the Gaussian weight (and equals $c_d\ell^2 I$ with $c_d=\tfrac12$). The actual weights carry the cutoff $\chi$ as well as the Gaussian, but $\chi$ departs from $1$ only where $|p|\geq\alpha\ell$, on which the Gaussian is already $\leq e^{-\alpha^2}$; the cutoff therefore alters $\bar\Sigma$ and every sum below by a negligible relative $O(e^{-\alpha^2})$, which we suppress. It suffices to bound $\|\Sigma(x)-\bar\Sigma\|$, which the counting below does at a free cover scale; \eqref{eq:sigma-bound} is then immediate from Weyl's inequality.

\emph{Counting with multiplicity.}
Each active point lies in exactly $M(p_k)$ of the diamonds, so for any function $g$ on $\R^d$,
\[
   \sum_j\ \sum_{p_k\in D_j} g(p_k)=\sum_k M(p_k)\,g(p_k).
\]
Take $g(p)=w(p)\,pp^T$. Inserting $M(p_k)=\kappa_d\varepsilon^{-d}(1+O(\varepsilon+s/\ell))$ from Step~2,
\[
   \sum_k w_k\,p_kp_k^T
   =\frac{\varepsilon^d}{\kappa_d}\sum_j\ \sum_{p_k\in D_j} w_k\,p_kp_k^T
   \cdot\bigl(1+O(\varepsilon+s/\ell)\bigr).
\]
Within a diamond of diameter $s$ the matrix $w\,pp^T$ varies by a relative $O(s/\ell)$ (its gradient is $O(\ell)$ on the support $|p|\lesssim\ell$), while $N_j=\rho\kappa_d s^d(1\pm\delta_s)$ by F2; hence $\sum_{p_k\in D_j} w_k\,p_kp_k^T=w(c_j)\,c_jc_j^T\,N_j\,(1+O(s/\ell))$. The centers grid-sample $\R^d$ with cell volume $(\varepsilon s)^d$, so $\sum_j w(c_j)c_jc_j^T=(\varepsilon s)^{-d}\int w(p)\,pp^T\,dp\,(1+O(\varepsilon))$. The geometric constants cancel exactly,
\[
   \frac{\varepsilon^d}{\kappa_d}\cdot\rho\kappa_d s^d\cdot(\varepsilon s)^{-d}=\rho,
\]
so $\sum_k w_k\,p_kp_k^T=\rho\int w\,pp^T\,dp\, (1+O(\varepsilon+s/\ell+\delta_s))$. The identical computation with $g=w$ gives the same factor for $W=\sum_k w_k$, and dividing yields
\[
   \bigl\|\Sigma(x)-\bar\Sigma\bigr\|
   \leq C_d\,\ell^2\bigl(\varepsilon+s/\ell+\delta_s\bigr).
\]
The left-hand side is independent of the grid parameter $\varepsilon$, so letting $\varepsilon\to0$ gives, for every admissible cover scale $s\in[\tau_{\min},\ell]$,
\begin{equation}\label{eq:sigma-free-s}
   \bigl\|\Sigma(x)-\bar\Sigma\bigr\|
   \leq C_d\,\ell^2\bigl(\delta_s+s/\ell\bigr),\qquad c_d=\tfrac12.
\end{equation}

\emph{Step~4: Optimizing the cover scale.} Write $L:=\log\rho V_M$, so $\delta_s=K_d\sqrt{L/(\rho\kappa_d s^d)}=\Theta\bigl((L/\rho)^{1/2}s^{-d/2}\bigr)$. The two terms in~\eqref{eq:sigma-free-s} trade off: a coarser cover lowers the count tolerance $\delta_s$ but raises the within-diamond variation $s/\ell$. Minimizing $f(s):=\delta_s+s/\ell$, the stationarity condition $f'(s)=0$ gives $s_*^{\,d/2+1}=\Theta(\ell\sqrt{L/\rho})$, that is $s_*=\Theta\bigl((\ell\sqrt{L/\rho})^{2/(d+2)}\bigr)$, at which the two terms balance to
\[
   f(s_*)=\Theta\Bigl(\bigl(L/(\rho\ell^d)\bigr)^{1/(d+2)}\Bigr)=\Theta(\delta_{\mathrm{cov}}),
\]
with $\delta_{\mathrm{cov}}$ as in~\eqref{eq:delta-cov}. The optimum is admissible: $s_*\leq\ell\iff\rho\ell^d\geq L$, automatic since $\rho\ell^d\gg1$; and $s_*\geq\tau_{\min}$ holds for $\rho\lambda^d$ large, a positive power of $\rho$ beating $L$ (for $d=4$ at the smoothing scale $\ell=\ell_*$ of Theorem~\ref{thm:haupt}, $s_*\sim\rho^{-0.18}$ lies strictly between $\tau_{\min}\sim\rho^{-1/4}$ and $\ell_*\sim\rho^{-1/20}$). Substituting $s=s_*$ into~\eqref{eq:sigma-free-s} gives~\eqref{eq:sigma-iso}; the lower bound~\eqref{eq:sigma-bound} follows by Weyl's inequality. Every constant is dimensional and the F2 bounds hold uniformly over admissible diamonds, so the estimate is uniform over $x\in M^\circ$.
\end{proof}

\subsection{The moving Karcher mean}\label{subsec:karcher}

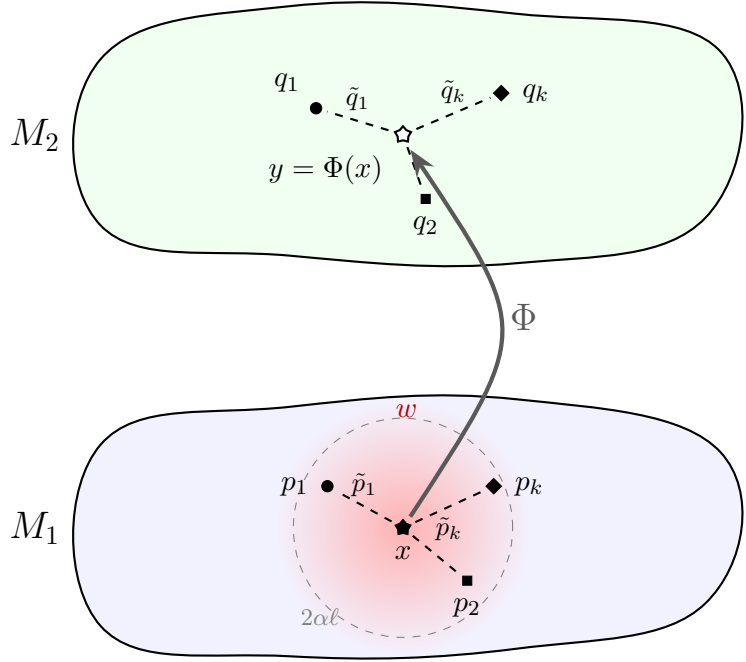
\begin{figure}[h!]
    \centering
    \begin{tikzpicture}[
        >=Stealth,
        centerpt/.style ={star, star points=5, draw=black, fill=black, inner sep=1.5pt},
        centeropen/.style={star, star points=5, draw=black, fill=white, thick, inner sep=1.5pt},
        corrA/.style={circle,    draw=black, fill=black, inner sep=1.5pt},
        corrB/.style={rectangle, draw=black, fill=black, inner sep=1.7pt},
        corrK/.style={diamond,   draw=black, fill=black, inner sep=1.5pt},
        connect/.style={dashed, thick, shorten >=2pt, shorten <=2pt}
    ]

    \begin{scope}[shift={(0,0)}]
        \filldraw[fill=blue!5, draw=black, thick]
            plot [smooth cycle, tension=0.7] coordinates {
            (-4,-1.3) (-1.5,-1.7) (1.5,-1.6) (4,-1.0)
            (4.2,1.0) (1.5,1.7) (-1.5,1.6) (-4,1.1) };
        \node[left] at (-4.4, 0) {\Large $M_1$};

        \shade[inner color=red!30, outer color=blue!5] (0,0) circle (1.7);
        \draw[dashed, gray] (0,0) circle (1.45);
        \node[gray, font=\footnotesize] at (-1.1,-1.18) {$2\alpha\ell$};
        \node[red!70!black] at (0.05,1.55) {$w$};

        \node[centerpt, label={[label distance=1pt]below:{$x$}}] (x) at (0,0) {};
        \node[corrA, label={[label distance=1pt]left:{$p_1$}}]  (p1) at (-1.0, 0.55) {};
        \node[corrB, label={[label distance=1pt]below:{$p_2$}}] (p2) at ( 0.85,-0.7) {};
        \node[corrK, label={[label distance=1pt]right:{$p_k$}}] (pk) at ( 1.2, 0.55) {};

        \draw[connect] (x) -- (p1) node[midway, above, font=\small] {$\tilde{p}_1$};
        \draw[connect] (x) -- (p2);
        \draw[connect] (x) -- (pk) node[midway, below, font=\small] {$\tilde{p}_k$};
    \end{scope}

    \begin{scope}[shift={(0, 5.2)}]
        \filldraw[fill=green!6, draw=black, thick]
            plot [smooth cycle, tension=0.7] coordinates {
            (-4,-1.2) (-1.5,-1.6) (1.5,-1.7) (4,-1.1)
            (4.2,1.1) (1.5,1.6) (-1.5,1.7) (-4,1.0) };
        \node[left] at (-4.4, 0) {\Large $M_2$};

        \node[centeropen, label={[label distance=2pt]below left:{$y=\Phi(x)$}}] (y) at (0,0) {};
        \node[corrA, label={[label distance=1pt]above left:{$q_1$}}] (q1) at (-1.15,0.35) {};
        \node[corrB, label={[label distance=1pt]below:{$q_2$}}]      (q2) at ( 0.30,-0.85) {};
        \node[corrK, label={[label distance=1pt]right:{$q_k$}}]      (qk) at ( 1.3, 0.55) {};

        \draw[connect] (y) -- (q1) node[midway, above, font=\small] {$\tilde{q}_1$};
        \draw[connect] (y) -- (q2);
        \draw[connect] (y) -- (qk) node[midway, above, font=\small] {$\tilde{q}_k$};
    \end{scope}

    \draw[->, ultra thick, color=black!65, shorten >=5pt, shorten <=5pt]
        (0,0) .. controls (1.7,2.6) .. (0.0,5.15)
        node[pos=0.55, right, font=\Large] {$\Phi$};

    \end{tikzpicture}
    \caption{Construction of the map $\Phi$. On the source $M_1$, a smooth mesoscopic weight $w$ (shaded) centered at $x$ singles out the active ball of radius $2\alpha\ell$ and the points $p_k$ within it, recorded in normal coordinates $\tilde{p}_k=\exp_x^{-1}(p_k)$. Their images $q_k\in M_2$ form an approximately Lorentz-transformed configuration. The value $\Phi(x)=y$ is the weighted Riemannian center of mass of the $q_k$: the unique point at which the geometric tension $F(x,y)=0$. Matching marker shapes (circle, square, diamond) denote corresponding elements of the causal set.}
    \label{fig:karcher}
\end{figure}

Let $p_k:=f_1(c_k)\in M_1$ and $q_k:=f_2(c_k)\in M_2$ denote the embedded points; the construction below is illustrated in Figure~\ref{fig:karcher}.

\begin{construction}[The diffeomorphism $\Phi$]
\label{con:phi}
Fix a smooth, non-negative, non-increasing bump function $\chi\colon\R\to[0,1]$ with $\chi(s)=1$ for $s\leq 1$ and $\chi(s)=0$ for $s\geq 2$. Fix a \emph{cutoff multiplier} $\alpha\geq 16$ with $2\alpha\ell\ll\mathrm{inj}(M_i,h_i)$ for $i=1,2$ (which holds since $\ell\ll\lambda$ and $\mathrm{inj}\geq c\lambda$ by Remark~\ref{rem:aux-curv}). Define \emph{cutoff Gaussian weights}
\begin{equation}\label{eq:weights}
   w_k(x) :=
   \exp\!\bigl(-d_{h_1}^2(x,p_k)/\ell^2\bigr)\,
   \chi\!\bigl(d_{h_1}(x,p_k)/(\alpha\ell)\bigr),
\end{equation}
and the \emph{moving Karcher mean}
\begin{equation}\label{eq:karcher}
   \Phi(x) := \argmin_{y\in M_2}\; \sum_k w_k(x)\;d_{h_2}^2(y,q_k).
\end{equation}
Each weight $w_k(x)$ is supported in the ball $B_{2\alpha\ell}^{h_1}(p_k)$, equals $\exp(-d_{h_1}^2(x,p_k)/\ell^2)$ on $B_{\alpha\ell}^{h_1}(p_k)$, and is smoothly tapered to zero in the annulus $\alpha\ell\leq d_{h_1}(x,p_k)\leq 2\alpha\ell$. Because $2\alpha\ell$ lies strictly inside the injectivity radius of $h_1$, the distance function $d_{h_1}(x,p_k)$ is $C^\infty$ on the support of $w_k$, hence $w_k\in C^\infty(M_1)$.

The cutoff truncation introduces only superexponentially small modifications relative to the pure Gaussian weight: the difference between $\exp(-d_{h_1}^2/\ell^2)$ and $w_k$ is bounded by $e^{-\alpha^2}$ pointwise (on the annulus $\alpha\ell\leq|p|\leq 2\alpha\ell$ where $\chi'\neq 0$), and integrates to a contribution $O(\alpha^{d-2}e^{-\alpha^2})$ relative to the total weight. For any fixed dimension $d$, choosing $\alpha\geq 16$ makes this correction smaller than every other error term in the proof (the volume-tolerance $O(\delta_{\mathrm{cov}})$, the curvature $O(\ell^2/\lambda^2)$, and the Bollob\'as--Brightwell longest-chain rate $O((\rho V)^{-1/(2d)})$). The dimensional prefactor $\alpha^{d-2}$ grows only polynomially in $\alpha$, while the suppression $e^{-\alpha^2}$ is double-exponential; thus $\alpha=16$ suffices uniformly in all dimensions of interest.
\end{construction}

\begin{remark}[Role of the cutoff weights]
\label{rem:cutoff-weights}
The cutoff serves two purposes. First, the support of $w_k$ is compact and strictly inside the injectivity radius, so the Karcher mean is well-defined by standard theorems~\cite{Karcher1977,Kendall1990}. Second, it sharply separates the inner core $\{|p_k|\leq r_d\ell\}$, where the trilateration argument applies and which must dominate the empirical cross-covariance, from the boundary layer, where the Gaussian weights are exponentially small (quantified in the error-term bound of Lemma~\ref{lem:cross-cov}).
\end{remark}

\begin{remark}[Curvature of the auxiliary metric]
\label{rem:aux-curv}
With $T$ chosen to (approximately) achieve the supremum in~\eqref{eq:lambda-intrinsic}, the Lorentzian curvature scale satisfies $|\mathrm{Rm}[g]|\leq 1/\lambda^2$ and the auxiliary injectivity radius satisfies $\mathrm{inj}(M,h_T)\geq\lambda$. The auxiliary Riemannian metric $h_T=g+2T^\flat\otimes T^\flat$ differs from $g$ by a smooth rank-$1$ perturbation. The Riemann tensor of $h_T$ satisfies $|\mathrm{Rm}[h_T]| =O(|\mathrm{Rm}[g]|+|\nabla T|^2+|\nabla^2 T|)$. For the optimizing $T$, $|\nabla T|=O(1/\lambda)$ and $|\nabla^2 T|=O(1/\lambda^2)$, giving $|\mathrm{Rm}[h_T]|=O(1/\lambda^2)$. Hence balls of radius $\ell\ll\lambda$ in $(M,h_T)$ are strongly convex. All subsequent estimates use this optimizing $T$ and the resulting intrinsic $\lambda$.
\end{remark}

\begin{proposition}[Well-definedness of the Karcher mean]
\label{prop:well-def}
For $\rho$ sufficiently large, the minimizer in~\eqref{eq:karcher} exists and is unique; hence $\Phi(x)$ is well-defined for every $x\in M_1^\circ$.
\end{proposition}
\begin{proof}
With cutoff weights, only finitely many $w_k(x)$ are nonzero at any given $x$ (those $p_k$ with $d_{h_1}(x,p_k)<2\alpha\ell$). By the target-localization Corollary~\ref{cor:target-localization}, the active target points $\{q_k:w_k(x)>0\}$ are all contained in a single $h_2$-ball $B_{R_*}^{h_2}(q_0)$ of radius $R_*=C_d'\alpha\ell\ll\mathrm{inj}(M_2,h_2)$, which is strongly convex (Remark~\ref{rem:aux-curv}). On this ball, $d_{h_2}^2(y,q_k)$ is strictly convex in $y$ for each fixed $q_k$ (standard Riemannian comparison; the ball is inside the injectivity radius and below the curvature scale). The weighted sum $\sum_k w_k(x)d_{h_2}^2(y,q_k)$ is therefore strictly convex on $B_{R_*}^{h_2}(q_0)$, with a unique minimum. Outside this ball the objective grows (each term increases as $y$ recedes from all active $q_k$), so the global minimum lies inside. Hence the Karcher mean $\Phi(x)$ is well-defined and unique.

The localization uses only the trilateration scaffold (Lemma~\ref{lem:scaffold} and Lemmas~\ref{lem:finite-procrustes}, \ref{lem:trilat-extension}), which depends on the embeddings and the preserved timelike proper times but not on $\Phi$, so invoking it here is well-founded.
\end{proof}

\begin{proposition}[Smoothness and derivative bounds]\label{prop:basic}
For $\rho$ sufficiently large:
\begin{enumerate}[(i)]
\item \textbf{Well-defined:} The minimizer in~\eqref{eq:karcher} exists and is unique (Proposition~\ref{prop:well-def}).
\item \textbf{Smooth:} $\Phi\in C^\infty(M_1,M_2)$.
\item \textbf{Approximate point matching:} $d_{h_2}(\Phi(p_j),q_j)=O\!\bigl((\varepsilon_\tau+\delta_{\mathrm{cov}}+\ell^2/\lambda^2)\,\ell\bigr)$ (Eq.~\eqref{eq:point-match}).
\item \textbf{Derivative bounds:} $\|D\Phi\|_{C^0}=O(1)$ and $\|D^2\Phi\|_{C^0}=O(1/\ell)$.
\end{enumerate}
\end{proposition}
\begin{proof}
\textit{Part~(ii): Smoothness.}
The minimizer satisfies the first-order condition
\begin{equation}\label{eq:foc}
   F(x,y) := \sum_k w_k(x)\,\exp_y^{-1}(q_k) = 0,
\end{equation}
where $\exp_y^{-1}(q_k)=-\tfrac{1}{2}\nabla_y d^2(y,q_k)$. The Jacobian $D_yF$ is the Hessian of the objective, which is strictly positive-definite (Proposition~\ref{prop:well-def}).

Each $w_k$ is $C^\infty$ globally on $M_1$: on the support of $w_k$ (where $d_{h_1}(x,p_k)<2\alpha\ell<\mathrm{inj}$), the distance function and the bump function are smooth, and outside the support $w_k\equiv 0$. The exponential map $\exp_y^{-1}$ is $C^\infty$ within the injectivity radius of $h_2$, which contains all active $q_k$ for $y$ near the minimum. By the implicit function theorem, the solution $y=\Phi(x)$ of $F(x,\Phi(x))=0$ is $C^\infty$ in $x$.

\textit{Part~(iii): Approximate point matching.}
Fix a base element $j$ with source point $p_j$ and target point $q_j$, and work in $h_1$-normal coordinates centered at $p_j$ and $h_2$-normal coordinates centered at $q_j$. Write the source and target weighted-mean displacements
\[
   \bar\Delta:=\frac1W\sum_k w_k(p_j)\,(p_k-p_j),
   \qquad
   \bar\Delta':=\frac1W\sum_k w_k(p_j)\,(q_k-q_j),
   \qquad W:=\sum_k w_k(p_j).
\]
The first-order condition~\eqref{eq:foc} for the Karcher mean gives $\Phi(p_j)=q_j$ to within $O(\|\bar\Delta'\|+\ell^2/\lambda^2\cdot\ell)$, the second term being the curvature correction to the linearized barycenter (the exponential map differs from its linearization by $O(\ell^2/\lambda^2)$ on the active ball). It therefore suffices to bound $\bar\Delta'$.

\emph{The source mean is deterministically small.}
By the dense admissible cover of Lemma~\ref{lem:source-cov} at a free scale $s\in[\tau_{\min},\ell]$ (overlapping diamonds of proper-time height $s$, centers on $\varepsilon s\,\Z^d$, uniform multiplicity $M=\kappa_d\varepsilon^{-d}(1+O(\varepsilon+s/\ell))$), every active point is counted, $\sum_k g(p_k)=\frac1M\sum_j\sum_{p_k\in D_j}g(p_k)$. Working in the $p_j$-centered coordinates ($p_j=0$), take $g(p)=w(p)\,p$, write each $p_k=c_j+r_k$ ($|r_k|\leq s$), and use $N_j=\rho\kappa_d s^d(1\pm\delta_s)$ by~(F2) together with $N_j/M=\rho(\varepsilon s)^d$:
\[
   \sum_k w_k\,p_k
   =\underbrace{\rho(\varepsilon s)^d\sum_j w(c_j)\,c_j}_{(\mathrm{I})}
   +\underbrace{\tfrac1M\sum_j w(c_j)\!\!\sum_{p_k\in D_j}\!\! r_k}_{(\mathrm{II})}
   +(\text{weight-variation terms}).
\]
Term~(I) is the Riemann sum of $\rho\int p\,e^{-|p|^2/\ell^2}\,dp=0$ and vanishes to leading order by central symmetry: the grid $\varepsilon s\,\Z^d$ and the Gaussian weight $w(c_j)=e^{-|c_j|^2/\ell^2}$ are both even under $c_j\mapsto-c_j$, so the symmetric partner of each $c_j$ cancels it; the residual comes only from the $O(\delta_s)$ asymmetry that (F2) permits between $N_j$ and $N_{-j}$, giving $|(\mathrm{I})|\leq C_d\,\delta_s\,W\ell$.
Term~(II) is bounded by $\tfrac1M\sum_j w(c_j)N_j\,s= W s =O(s/\ell)\,W\ell$ (using $\tfrac1M\sum_j w(c_j)N_j=\sum_k w_k=W$). Hence $\|\bar\Delta\|\leq C_d(\delta_s+s/\ell)\ell$ for every $s\in[\tau_{\min},\ell]$; optimizing the cover scale exactly as in Step~4 of Lemma~\ref{lem:source-cov} gives
\begin{equation}\label{eq:source-mean}
   \|\bar\Delta\|\leq C_d\,\delta_{\mathrm{cov}}\,\ell.
\end{equation}

\emph{Transfer to the target mean.}
By the target-localization Corollary~\ref{cor:target-localization}, every active point satisfies $q_k-q_j=\hat\Lambda(p_k-p_j)+e_k$ with $\|e_k\|\leq C_d\varepsilon_\tau\alpha\ell$ (here both displacements are measured from the matched base points $p_j,q_j$, and $\hat\Lambda$ is the local Lorentz map of Lemma~\ref{lem:finite-procrustes}). Averaging,
\[
   \bar\Delta'=\hat\Lambda\,\bar\Delta+\frac1W\sum_k w_k e_k,
   \qquad
   \Bigl\|\frac1W\sum_k w_k e_k\Bigr\|\leq C_d\varepsilon_\tau\alpha\ell .
\]
Using $\|\hat\Lambda\|=O(1)$ (Lemma~\ref{lem:finite-procrustes}, \eqref{eq:proc-boost}) and~\eqref{eq:source-mean},
\[
   \|\bar\Delta'\|
   \leq C_d\bigl(\delta_{\mathrm{cov}}+\varepsilon_\tau\bigr)\ell .
\]
Combining with the curvature term,
\begin{equation}\label{eq:point-match}
   d_{h_2}\bigl(\Phi(p_j),q_j\bigr)
   \leq C_d\bigl(\varepsilon_\tau+\delta_{\mathrm{cov}}
   +\ell^2/\lambda^2\bigr)\ell,
\end{equation}
which is the deterministic matching residual; all three error terms vanish as $\rho\lambda^d\to\infty$ under scale separation. No appeal to the random embedding is made: the cancellation in~(I) is a consequence of the grid symmetry forced by~(F2), exactly as the isotropy in Lemma~\ref{lem:source-cov}, while the transfer to the target is supplied by the deterministic trilateration of Section~\ref{sec:trilat}.

\textit{Part~(iv): Derivative bounds.}
Differentiating~\eqref{eq:foc} implicitly:
\[
   D_xF + D_yF\cdot D\Phi = 0
   \;\;\Longrightarrow\;\;
   D\Phi = -[D_yF]^{-1}\,D_xF.
\]

\textit{Scale of $D_yF$:} The weights $w_k(x)$ depend only on $x$, not on $y$, so the $y$-derivative hits only the inverse exponential map: $D_yF=\sum_k w_k\,\partial_y[\exp_y^{-1}(q_k)]$, where $\partial_y$ is the derivative with respect to the basepoint $y$ at fixed target $q_k$. In flat space, $\exp_y^{-1}(q)=q-y$, so $\partial_y[\exp_y^{-1}(q)]=-I$. In curved space, the Jacobi-field expansion of the exponential map gives $\partial_y[\exp_y^{-1}(q_k)]=-I+O(|\tilde w_k'|^2/\lambda^2)$, where the curvature correction is bounded by $\|\mathrm{Rm}\|\cdot|\tilde w_k'|^2=O(\ell^2/\lambda^2)$. Summing, $D_yF=-W[I+O(\ell^2/\lambda^2)]$ where $W:=\sum_k w_k$. Hence $[D_yF]^{-1}=-(1/W)[I+O(\ell^2/\lambda^2)]$.

\textit{Scale of $D_xF$:} With the cutoff weights $w_k(x)=e^{-d^2/\ell^2}\chi(d/(\alpha\ell))$ where $d:=d_{h_1}(x,p_k)$, the gradient is
\[
   \nabla_x w_k = -\frac{2}{\ell^2}\,e^{-d^2/\ell^2}
   \chi(d/\alpha\ell)\nabla_x(d^2/2)
   + \frac{1}{\alpha\ell}e^{-d^2/\ell^2}
   \chi'(d/\alpha\ell)\,\nabla_x d.
\]
The first term equals $+\frac{2}{\ell^2}\tilde w_k\,w_k$ on the inner support $d\leq\alpha\ell$ (where $\chi=1$), with $\tilde w_k=\exp_x^{-1}(p_k)$.

The second term is supported in the annulus $\alpha\ell\leq d\leq 2\alpha\ell$. On this annulus: $e^{-d^2/\ell^2}\leq e^{-\alpha^2}$, $\chi'$ is bounded by $\|\chi'\|_\infty=O(1)$ (a dimensional constant of the fixed bump function), and $|\nabla_x d|=1$. Hence pointwise the second term is $O(e^{-\alpha^2}/(\alpha\ell))$.

However, the relevant quantity is not the pointwise bound but the contribution to the sum $D_xF=\sum_k\nabla_x w_k\otimes\exp_y^{-1}(q_k)$. We compare the annulus contribution to the main-term contribution by their typical scales. The main term is supported on the inner ball of volume $\sim(\alpha\ell)^d$ with $|p_k|,|q_k|\lesssim\alpha\ell$, giving (by the F2 count $\rho(\alpha\ell)^d(1\pm\delta_D)$ in the inner ball) $\|D_xF^{\mathrm{main}}\|\sim\rho\cdot(\alpha\ell)^d\cdot \frac{1}{\ell^2}\cdot(\alpha\ell)^2=O(\rho\alpha^{d+2}\ell^d) =O(W)$, where $W=\sum_k w_k\sim\rho\ell^d$. The annulus term has the same volume scale $(\alpha\ell)^d$, the same displacement scale $\alpha\ell$, but is multiplied by the pointwise bound $e^{-\alpha^2}/(\alpha\ell)$, giving
\[
   \|D_xF^{\mathrm{annulus}}\|
   \leq\rho\cdot(\alpha\ell)^d\cdot\frac{e^{-\alpha^2}}{\alpha\ell}
   \cdot\alpha\ell=\rho\,\alpha^d\ell^d\,e^{-\alpha^2}.
\]
The relative correction is therefore
\begin{equation}\label{eq:annulus-ratio}
   \frac{\|D_xF^{\mathrm{annulus}}\|}{\|D_xF^{\mathrm{main}}\|}
   \;\lesssim\;\alpha^d e^{-\alpha^2},
\end{equation}
a dimensionless constant depending only on $\alpha$ and $d$. For $\alpha=16$, $d=4$: $\alpha^d e^{-\alpha^2}=16^4\cdot e^{-256}\approx 10^{-107}$. This ratio is independent of $\rho$ and $\ell$: as $\rho\to\infty$ and $\ell\to 0$, both the main and annulus contributions scale as $\rho\ell^d$, so the correction remains superexponentially suppressed relative to the main term.

Modulo this superexponentially small relative correction, component-wise $\partial_{x^j}F^i = +(2/\ell^2)\sum_k w_k\,\tilde w_k'^i\,\tilde w_k^j$, so:
\[
   D_xF = +\frac{2}{\ell^2}\sum_k w_k\,
   \tilde w_k'\otimes\tilde w_k \cdot\bigl(1+O(\alpha^d e^{-\alpha^2})\bigr).
\]
This is $\frac{2W}{\ell^2}\Sigma_{\mathrm{cross}}$ where $\Sigma_{\mathrm{cross}}=\frac{1}{W}\sum_k w_k\, \tilde w_k'\,\tilde w_k^T$ is the weighted cross-covariance of source and target displacements. The Gaussian weight $e^{-|\tilde w_k|^2/\ell^2}$ concentrates support on $|\tilde w_k|\lesssim\ell$ (the cutoff at $\alpha\ell$ is far in the Gaussian tail), so typical displacements are $O(\ell)$ and $\|\Sigma_{\mathrm{cross}}\|=O(\ell^2)$. Hence $\|D_xF\|=O(W)$.

\textit{Combining:}
$\|D\Phi\|=\|[D_yF]^{-1}\|\cdot\|D_xF\| =O(1/W)\cdot O(W)=O(1)$.

For the second derivative, differentiate $D_xF+D_yF\cdot D\Phi=0$ once more in $x$ (with $y=\Phi(x)$):
\[
   D^2\Phi=-[D_yF]^{-1}\bigl(D_{xx}F+2\,D_{xy}F\cdot D\Phi +D_{yy}F\cdot(D\Phi\otimes D\Phi)\bigr).
\]
The Hessian blocks scale as $\|D_{xx}F\|=O(W/\ell)$ (two $x$-derivatives of the Gaussian weight, $D^2w_k\sim w_k/\ell^2$, against displacements $|\exp_y^{-1}(q_k)|=O(\ell)$), $\|D_{xy}F\|=O(W/\ell)$ (one weight derivative against one exponential-map derivative), and $\|D_{yy}F\|=O(W\ell/\lambda^2)$ (two $y$-derivatives of $\exp_y^{-1}$, a curvature term $O(1/\lambda^2)$ against displacements $O(\ell)$). With $\|[D_yF]^{-1}\|=O(1/W)$ and $\|D\Phi\|=O(1)$,
\[
   \|D^2\Phi\|
   =O(1/W)\bigl[O(W/\ell)+O(W/\ell)\cdot O(1)+O(W\ell/\lambda^2)\cdot O(1)\bigr]
   =O(1/\ell),
\]
the first two blocks dominating since $\ell\ll\lambda$. Equivalently, from $D\Phi=(2/\ell^2)\Sigma_{\mathrm{cross}}+O(\ell^2/\lambda^2)$ with $\|\Sigma_{\mathrm{cross}}\|=O(\ell^2)$ varying on the smoothing scale $\ell$, one has $\|D\Sigma_{\mathrm{cross}}\|=O(\ell)$ and hence $\|D^2\Phi\|=O(1/\ell)$ directly.
\end{proof}

\begin{lemma}[Cross-covariance non-degeneracy]
\label{lem:cross-cov}
Under the well-conditioning hypotheses (F1)--(F3), for every $x\in M_1^\circ$ the cross-covariance
\[
   \Sigma_{\mathrm{cross}}(x)
   :=\frac{1}{W}\sum_k w_k(x)\,\tilde w_k'\,\tilde w_k^T,
   \qquad
   \tilde w_k:=\exp_x^{-1}(p_k),\quad
   \tilde w_k':=\exp_{\Phi(x)}^{-1}(q_k),
\]
satisfies
\begin{equation}\label{eq:crosscov-bound}
   \sigma_{\min}\bigl(\Sigma_{\mathrm{cross}}(x)\bigr) \;\geq\; c_\Lambda\,c_d\,\ell^2 \Bigl(1-C_d(\varepsilon_\tau+\delta_{\mathrm{cov}})\Bigr) \;\geq\;\tfrac12 c_\Lambda c_d\,\ell^2,
\end{equation}
uniformly over $x\in M_1^\circ$, for $\rho\lambda^d$ sufficiently large.
\end{lemma}

\begin{proof}
Work in $h_1$-normal coordinates at $x$ and $h_2$-normal coordinates at $\Phi(x)$. By the target-localization Corollary~\ref{cor:target-localization} (whose hypotheses are supplied by F3 for distance preservation and, for the conditioned scaffold, Lemma~\ref{lem:scaffold}, itself a consequence of F2), there is a Lorentz transformation $\hat\Lambda\in O(1,d-1)$, uniformly conditioned with $\sigma_{\min}(\hat\Lambda)\geq c_\Lambda$ and $\|\hat\Lambda\|\leq C_\Lambda$ (Lemma~\ref{lem:finite-procrustes}, \eqref{eq:proc-boost}). Although Lemma~\ref{lem:finite-procrustes} fixes $\hat\Lambda$ from the $d+1$ scaffold anchors alone, the trilateration extension (Lemma~\ref{lem:trilat-extension}, applied in Corollary~\ref{cor:target-localization}) propagates the resulting affine relation to every active point, not merely the anchors, so that
\begin{equation}\label{eq:slot4-loc}
   \tilde w_k'=\hat\Lambda\,\tilde w_k+e_k,
   \qquad \|e_k\|\leq C_d\,\varepsilon_\tau\,\alpha\ell.
\end{equation}
Substituting into the definition of $\Sigma_{\mathrm{cross}}$,
\begin{equation}\label{eq:slot4-split}
   \Sigma_{\mathrm{cross}}
   =\frac{1}{W}\sum_k w_k(\hat\Lambda\tilde w_k+e_k)\tilde w_k^T
   =\hat\Lambda\,\Sigma_1+E,
   \qquad
   E:=\frac{1}{W}\sum_k w_k\,e_k\,\tilde w_k^T,
\end{equation}
where $\Sigma_1=\frac1W\sum_k w_k\tilde w_k\tilde w_k^T$ is the source covariance of Lemma~\ref{lem:source-cov}.

\emph{The error term.}
Using \eqref{eq:slot4-loc} and $\|\tilde w_k\|\leq 2\alpha\ell$ on the cutoff support,
\[
   \|E\|
   \leq\frac{1}{W}\sum_k w_k\,\|e_k\|\,\|\tilde w_k\|
   \leq C_d\varepsilon_\tau\alpha\ell\cdot 2\alpha\ell
   \cdot\frac{1}{W}\sum_k w_k
   = O(\varepsilon_\tau)\,\ell^2,
\]
since $\alpha=O(1)$ and $\frac1W\sum_k w_k=1$.

\emph{Lower bound.}
For square matrices, $\sigma_{\min}(AB)\geq \sigma_{\min}(A)\sigma_{\min}(B)$ (take $\min_{\|u\|=1}\|ABu\|\geq\sigma_{\min}(A)\min_{\|u\|=1}\|Bu\|$). Hence, by Weyl's perturbation inequality applied to~\eqref{eq:slot4-split},
\[
   \sigma_{\min}(\Sigma_{\mathrm{cross}})
   \geq\sigma_{\min}(\hat\Lambda\,\Sigma_1)-\|E\|
   \geq\sigma_{\min}(\hat\Lambda)\,\sigma_{\min}(\Sigma_1)-\|E\|
   \geq c_\Lambda\,\sigma_{\min}(\Sigma_1)-O(\varepsilon_\tau)\ell^2.
\]
By Lemma~\ref{lem:source-cov}, $\sigma_{\min}(\Sigma_1)\geq c_d\ell^2(1-C_d\,\delta_{\mathrm{cov}})$. Combining,
\[
   \sigma_{\min}(\Sigma_{\mathrm{cross}})
   \geq c_\Lambda c_d\,\ell^2
   \Bigl(1-C_d(\varepsilon_\tau+\delta_{\mathrm{cov}})\Bigr),
\]
which is~\eqref{eq:crosscov-bound}. Since the admissible
range $[\tau_{\min},c_*\lambda]$ is non-empty, all three
error terms vanish as $\rho\lambda^d\to\infty$, so the
bracket exceeds $\tfrac12$ for $\rho\lambda^d$ large.
Every constant is dimensional and the inputs hold uniformly
over $M_1^\circ$, so the bound is uniform.
\end{proof}

\begin{proposition}[$\Phi$ is a local diffeomorphism]
\label{prop:local-diffeo}
For $\rho$ sufficiently large, the differential $D\Phi(x)$ is invertible for every $x\in M_1^\circ$, with $\sigma_{\min}(D\Phi)\geq c_0>0$ for a dimensional constant $c_0$. In particular $\Phi$ is a local diffeomorphism on $M_1^\circ$.
\end{proposition}
\begin{proof}
By Proposition~\ref{prop:basic}(iv), the differential of the
first-order condition~\eqref{eq:foc} factors as
$D\Phi=-[D_yF]^{-1}D_xF$, with
\[
   D_yF=-W\bigl[I+O(\ell^2/\lambda^2)\bigr],
   \qquad
   D_xF=\frac{2W}{\ell^2}\,\Sigma_{\mathrm{cross}},
\]
where $W=\sum_k w_k$ and $\Sigma_{\mathrm{cross}}
=\frac1W\sum_k w_k\,\tilde w_k'\,\tilde w_k^T$ is the weighted
cross-covariance. Since $\ell\ll\lambda$, $D_yF$ is a small
perturbation of $-WI$, hence invertible with
$[D_yF]^{-1}=-(1/W)[I+O(\ell^2/\lambda^2)]$. By
Lemma~\ref{lem:cross-cov},
$\sigma_{\min}(\Sigma_{\mathrm{cross}})\geq c_2\ell^2$, so
$D_xF$ is invertible with $\sigma_{\min}(D_xF)\geq 2Wc_2$.
Therefore $D\Phi=-[D_yF]^{-1}D_xF$ is a product of invertible
operators, hence invertible, with
$\sigma_{\min}(D\Phi)\geq c_0>0$ for a dimensional constant
$c_0$. By the inverse function theorem, $\Phi$ is a local
diffeomorphism on $M_1^\circ$.
\end{proof}

\begin{proposition}[$\Phi$ is a global diffeomorphism]
\label{prop:diffeo}
For $\rho$ sufficiently large, $\Phi$ is a diffeomorphism from $M_1^\circ$ onto $M_2^\circ$.
\end{proposition}
\begin{proof}
By Proposition~\ref{prop:local-diffeo}, $D\Phi(x)$ is invertible for every
$x\in M_1^\circ$, with $\sigma_{\min}(D\Phi)\geq c_0>0$; we upgrade this local
diffeomorphism to a global one.

Throughout the global step we take the boundaryless case $\partial M_i=\emptyset$,
so that $M_i^\circ=M_i$ (Definition~\eqref{eq:M1circ-def} sets $M^\circ=M$ when
$\partial M=\emptyset$). The case $\partial M_i\neq\emptyset$ reduces to this by
precompact exhaustion and yields the regional statement of
Remark~\ref{rem:noncompact}, as discussed at the end of the proof; there the
conclusion is necessarily regional rather than a diffeomorphism onto all of $M_2$,
since a proper subdomain $M_1^\circ\subsetneq M_1$ cannot be diffeomorphic to a
complete $M_2$.

\emph{Constant sign of the Jacobian.}
Define the Jacobian
$J_\Phi\colon M_1^\circ\to\R$ by
$\Phi^*\mathrm{vol}_{h_2}=J_\Phi\cdot\mathrm{vol}_{h_1}$,
where $\mathrm{vol}_{h_i}$ are the Riemannian volume forms
of the auxiliary metrics. Although the numerical value of
$J_\Phi$ depends on the choice of auxiliary metrics $h_i$
(hence on the Cauchy time functions), its sign and
its non-vanishing are gauge-independent: both
$\mathrm{vol}_{h_1}$ and $\mathrm{vol}_{h_2}$ are positive
volume forms for any admissible choice, so $J_\Phi\neq 0$
iff $D\Phi$ is invertible, a gauge-free condition. We use
$J_\Phi$ only to detect local invertibility here; the
intrinsic isometry statement in
Section~\ref{sec:conformal} makes no reference to it. Since
$D\Phi$ is invertible everywhere
(Proposition~\ref{prop:local-diffeo}), $J_\Phi\neq 0$
everywhere. Since $J_\Phi$ is continuous on
the connected manifold $M_1^\circ$, it has constant sign
throughout. We do not at this stage determine which sign;
this only matters for orientation, which we discuss in
Remark~\ref{rem:orientation} below.

\emph{Domain restriction (geometric).}
With cutoff weights, $\Phi(x)$ is defined only when at
least one $w_k(x)>0$, i.e., when the ball
$B_{2\alpha\ell}^{h_1}(x)$ contains at least one embedded
point $p_k$. To obtain a domain on which $\Phi$ is well-defined, on which the Hadamard--Caccioppoli global inversion theorem (Krantz--Parks~\cite{KrantzParks2002}, Thm.~6.2.8) applies, and on which the F2
condition is available for all diamonds entering the proof,
we define $M_1^\circ$ purely geometrically,
independent of the random sprinkling:
\begin{equation}\label{eq:M1circ-def}
   M_1^\circ := \bigl\{x\in M_1 \;:\;
   d_{h_1}(x,\partial M_1) > c_*\lambda + 12\alpha\ell\bigr\}.
\end{equation}
This is a deterministic open subset of $M_1$ with smooth
boundary. The clearance $c_*\lambda+12\alpha\ell$ is chosen
so that three requirements are met simultaneously: (a) the
cutoff support $B_{2\alpha\ell}^{h_1}(x)$ of any
$x\in M_1^\circ$ lies wholly inside $M_1$ (needing
clearance $\geq 2\alpha\ell$); (b) the trilateration anchor
scaffold (Lemma~\ref{lem:scaffold}), whose anchors
sit at $h_1$-distance $\leq 8\sqrt2\,\alpha\ell+\varepsilon_\tau\alpha\ell
<12\alpha\ell$ from $x$, lies inside $M_1$; and, most
stringently, (c) any admissible causal diamond (of diameter up
to the maximal admissible scale $c_*\lambda$) centered within
$12\alpha\ell$ of $x$ satisfies the F2 boundary-clearance
requirement $d_{h_1}(D,\partial M_1)\geq c_*\lambda$ from
Definition~\ref{def:faithful}. Requirement (c) is the
binding one and is the reason the clearance is macroscopic
($\sim c_*\lambda$) rather than mesoscopic ($\sim\alpha\ell$).

\emph{Density on $M_1^\circ$.}
For every $x\in M_1^\circ$ the cutoff support $B_{2\alpha\ell}^{h_1}(x)$ contains
embedded points: it contains an admissible causal diamond of proper-time height
$\ell$ (admissible since $\tau_{\min}\leq\ell\leq c_*\lambda$ under scale
separation, and the $M_1^\circ$ clearance places it at $h_1$-distance
$\geq c_*\lambda$ from $\partial M_1$), to which (F2) applies, giving
$|f_1(C)\cap B_{2\alpha\ell}^{h_1}(x)|\geq\tfrac12\rho\,\mathrm{Vol}_{g_1}
\geq c_d\rho\ell^d>0$. Hence $\Phi(x)$ is well-defined for every
$x\in M_1^\circ$, with no appeal to the sprinkling law.

\emph{Target domain.}
There is no ambient space relating $M_1$ and $M_2$, so the
notion of a ``distance moved'' by $\Phi$ is undefined and we
do not use it. For the global assembly of Stage~B we work in the
boundaryless setting: either $\partial M_2=\emptyset$
(spatially compact Cauchy slices), or $M_2$ is treated by
precompact exhaustion (Remark~\ref{rem:noncompact}). In this
setting we set $M_2^\circ:=M_2$, so that $d_{h_2}(\cdot,
\partial M_2)=\infty$ and the inclusion $\Phi(M_1^\circ)
\subseteq M_2^\circ$ holds trivially; the target images
$\Phi(x)$ are confined to mesoscopic balls around embedded
points by the target-localization
Corollary~\ref{cor:target-localization}, all intrinsic to
$M_2$. The construction of $\Phi$ is thus carried out on the
deep interior $M_1^\circ$ of the source and all of the
(boundaryless) target, consistent with
Remark~\ref{rem:boundary-shrinks}.

\emph{Properness on $M_1^\circ$.}
We show that $\Phi\colon M_1^\circ\to M_2^\circ$ is proper:
the preimage of any compact set is compact. Let
$K\subset M_2^\circ$ be compact. We show
$\Phi^{-1}(K)\subset M_1^\circ$ is compact.

By the target-localization Corollary~\ref{cor:target-localization}, the active
target points $\{q_k:k\in S(x)\}$, $S(x):=\{k:w_k(x)>0\}$, and with them $\Phi(x)$,
all lie within $h_2$-distance $R_*=C_d'\alpha\ell$ of any one of them. Hence
$\Phi(x)\in K$ implies every active $q_k$ lies in the compact set
$K':=\{y\in M_2:d_{h_2}(y,K)\leq R_*\}$. Since $C$ is finite and $f_2$ is injective,
only finitely many embedded points have images in $K'$; let
$P:=f_1\bigl(f_2^{-1}(f_2(C)\cap K')\bigr)$ be the corresponding finite set of source
points and $K''$ the (compact) union of their closed $2\alpha\ell$-balls. Every $x$
with $\Phi(x)\in K$ has an active source point $p_k\in P$ within $h_1$-distance
$2\alpha\ell$, so $x\in K''$. Because $\partial M_1=\emptyset$ (standing hypothesis),
$\Phi^{-1}(K)$ is closed in $M_1$; being contained in the compact $K''$, it is
compact. This proves properness.

\emph{Hadamard--Caccioppoli.}
$\Phi$ is now a proper local diffeomorphism
$M_1^\circ\to M_2^\circ$ between connected manifolds ($M_i^\circ=M_i$,
connected because $M_i$ is). As a local diffeomorphism $\Phi$ is an
open map, and as a proper map into a locally compact Hausdorff
space it is a closed map; hence $\Phi(M_1^\circ)$ is a clopen
subset of the connected manifold $M_2^\circ$. It is nonempty
($\Phi$ maps each embedded $p_k$ to a nearby
$q_k\in M_2^\circ$), so $\Phi(M_1^\circ)=M_2^\circ$. By the
Hadamard--Caccioppoli theorem (loc.\ cit.), the proper local
diffeomorphism $\Phi$ is a covering map; being surjective onto
the connected $M_2^\circ$, it has a constant finite number of
sheets $k\in\mathbb{Z}_{>0}$.

\emph{Local injectivity.}
Fix $x_0\in M_1^\circ$ and let $y_0:=\Phi(x_0)$. By
Lemma~\ref{lem:cross-cov} and
Proposition~\ref{prop:local-diffeo}, on the inner-core ball
$B_{2\ell}^{h_1}(x_0)$ the map admits the trilateration
expansion $\Phi(x_0+v)=y_0+\hat\Lambda v+R(v)$ with
$\hat\Lambda\in O(1,d-1)$, $\sigma_{\min}(\hat\Lambda)\geq
c_\Lambda$, and $\|R(v)\|\leq\kappa|v|$, where $\kappa=o(1)$ is
the uniform derivative error $\|D\Phi-\hat\Lambda\|$ supplied by
Proposition~\ref{prop:basic}(iv) and Lemma~\ref{lem:cross-cov}
(both proved before this proposition). If $\Phi(x_0+v)=\Phi(x_0)$
with $|v|\leq 2\ell$, then $\hat\Lambda v=-R(v)$, so
$c_\Lambda|v|\leq\|\hat\Lambda v\|=\|R(v)\|\leq\kappa|v|$; once
$\rho$ is large enough that $\kappa<c_\Lambda$ this forces
$v=0$. Thus $\Phi$ is injective on every inner-core ball, a
deterministic bound requiring no non-cancellation argument.

\emph{Degree one via volume.}
It remains to show $k=1$. The sheet count cannot be pinned
by local data alone: $M_2^\circ$ need not be simply connected,
and a covering can be locally trivial yet globally multi-sheeted.
We instead compare volumes. For a $k$-sheeted covering the
change-of-variables formula gives
\begin{equation}\label{eq:degree-volume}
   \int_{M_1^\circ}|J_\Phi|\,\mathrm{vol}_{h_1}
   = k\,\mathrm{Vol}_{h_2}(M_2^\circ).
\end{equation}

\emph{(i) The Jacobian is uniformly close to $1$.} Measuring
$D\Phi$ in adapted $h_i$-normal coordinates (in which
$h_i=\delta$ and $g_i=\eta$ to leading order),
Lemma~\ref{lem:cross-cov} and
Proposition~\ref{prop:local-diffeo} give $D\Phi=\hat\Lambda+E$
with $\hat\Lambda\in O(1,d-1)$ and $\|E\|=o(1)$ uniformly on
$M_1^\circ$. Any $\hat\Lambda\in O(1,d-1)$ satisfies
$\hat\Lambda^T\eta\hat\Lambda=\eta$, so $\det(\hat\Lambda)^2=1$
and $|\det\hat\Lambda|=1$; by continuity of the determinant
$|J_\Phi|=|\det(\hat\Lambda+E)|=1+o(1)$ uniformly. Hence
$\int_{M_1^\circ}|J_\Phi|\,\mathrm{vol}_{h_1}
=\mathrm{Vol}_{h_1}(M_1^\circ)\,(1+o(1))$.

\emph{(ii) The auxiliary volume is the spacetime volume.} For
the auxiliary metric $h_T=g+2T^\flat\otimes T^\flat$ with $T$ a
$g$-unit timelike field, a $g$-orthonormal frame with $e_0=T$
gives $T^\flat=-e^0$, so $h_T=g+2\,e^0\otimes e^0=\delta$; both
volume densities equal $1$, whence
$\mathrm{vol}_{h_T}=\mathrm{vol}_g$ identically, for either
spacetime. Thus \eqref{eq:degree-volume} reads
$\mathrm{Vol}_{g_1}(M_1^\circ)\,(1+o(1))
=k\,\mathrm{Vol}_{g_2}(M_2^\circ)$.

\emph{(iii) The volumes match.} We argue in the boundaryless
setting $\partial M_i=\emptyset$, where $M_i^\circ=M_i$; the
general case follows by precompact exhaustion
(Remark~\ref{rem:noncompact}). Both embeddings present the same
finite causal set $C$ at common density $\rho$, so, tiling each
$M_i$ by admissible diamonds at the coarse scale $c_*\lambda$ and
summing the F2 counts, $\mathrm{Vol}_{g_1}(M_1)=|C|\rho^{-1}(1+o(1))
=\mathrm{Vol}_{g_2}(M_2)\,(1+o(1))$, with tolerance
$\delta_{c_*\lambda}$ polynomially small. Substituting,
$k=1+o(1)$; since $k$ is a positive integer, $k=1$ for $\rho$
sufficiently large. Hence $\Phi\colon M_1^\circ\to M_2^\circ$
is a single-sheeted covering, i.e., a diffeomorphism.

\emph{The boundary case.}
If $\partial M_i\neq\emptyset$, the deep interior
$M_1^\circ$ is a proper subdomain of $M_1$: the excluded layer
has width $c_*\lambda+12\alpha\ell$, whose macroscopic part
$c_*\lambda$ does not shrink as $\rho\to\infty$
(Remark~\ref{rem:boundary-shrinks}), so $\Phi\colon M_1^\circ\to
M_2$ is not proper onto all of $M_2$ and the covering argument
above does not apply globally. The conclusion is instead the
regional one of Remark~\ref{rem:noncompact}: for each precompact
region one obtains a diffeomorphism onto its image on the deep
interior, with the excluded layer receding only under
exhaustion. The onto-$M_2$ statement proved above is thus the
boundaryless case.
\end{proof}

\begin{remark}[Orientation]
\label{rem:orientation}
The conclusion ``$\Phi$ is a diffeomorphism'' allows
$\Phi$ to be either orientation-preserving or
orientation-reversing relative to fixed orientations on
$M_1$ and $M_2$. Both possibilities are consistent with the
faithful embedding hypotheses: a global spatial reflection
in $d\geq 4$ preserves Lorentzian distances, causal order
(F1), and volumes (F2), so two faithful embeddings related
by such a reflection are equally valid.
The Hauptvermutung therefore concludes that
$M_1$ and $M_2$ are approximately isometric as
unoriented Lorentzian manifolds. If one assumes in
addition that the two faithful embeddings are
orientation-compatible (a natural condition when both
spacetimes carry fixed physical orientations), the
diffeomorphism $\Phi$ is automatically
orientation-preserving: the local correspondence
$\tilde w_k\mapsto\tilde w_k'$ inherits orientation
compatibility, hence $\det(\Sigma_{\mathrm{cross}})>0$ at
each embedded point, hence $J_\Phi>0$ globally by the
constant sign of $J_\Phi$ established above.
\end{remark}

\begin{remark}[No exotic $\R^4$ obstruction]
\label{rem:exotic}
The construction is entirely local: at each $x$, the Karcher
mean uses only nearby point correspondences. No topological
classification of manifolds is invoked, so the proof works
uniformly in all dimensions including
$d=4$~\cite{Donaldson1987,Freedman1982}.
\end{remark}

\begin{remark}[Structure of the invertibility argument]
\label{rem:avoids}
The invertibility of $D\Phi$ rests on three ingredients, of
increasing scope. Volume-faithfulness~(F2) makes the source
covariance $\Sigma_1$ non-degenerate (Lemma~\ref{lem:source-cov});
the longest-chain distance preservation, organized through the
trilateration scaffold, makes the cross-covariance
$\Sigma_{\mathrm{cross}}$ non-degenerate; and the implicit function
theorem then gives the invertibility of $D\Phi$.
The combinatorial content enters only through preserved
timelike proper times: the trilateration locates
each point from its causal distances to a fixed
anchor scaffold, which determines an approximate Lorentz
transformation $\hat\Lambda$ relating the two local point
clouds. No spacelike distance is ever measured directly, and
the global metric comparison is deferred to the rigidity
argument of Section~\ref{sec:conformal}.
\end{remark}

\subsection{Rigidity}\label{sec:conformal}

At this stage $\Phi\colon M_1^\circ\to M_2^\circ$ is a diffeomorphism
(Proposition~\ref{prop:diffeo}). We now show it is an approximate isometry,
$\Phi^*g_2=g_1+E$ with $\|E\|$ small, by computing its differential from the
cross-covariance and using that the proper-time preservation~(F3) fixes the local
frame to an exact Lorentz transformation.

\begin{proposition}[Approximate isometry]
\label{prop:conformal}
The diffeomorphism $\Phi$ satisfies
$\Phi^*g_2=g_1+E$ with
$\|E\|_{g_1}=O\bigl(\varepsilon_\tau+\ell^2/\lambda^2
+\delta_{\mathrm{cov}}\bigr)$.
\end{proposition}
\begin{proof}
By Proposition~\ref{prop:basic}(iv),
$D\Phi=(2/\ell^2)\Sigma_{\mathrm{cross}}+O(\ell^2/\lambda^2)$, the
$O(\ell^2/\lambda^2)$ correction coming from the curvature of the exponential map in
$[D_yF]^{-1}$. By Lemma~\ref{lem:cross-cov},
$\Sigma_{\mathrm{cross}}=\hat\Lambda\,\Sigma_1+E_{\mathrm{cr}}$ with
$\hat\Lambda\in O(1,d-1)$ (Lemma~\ref{lem:finite-procrustes}) and
$\|E_{\mathrm{cr}}\|=O(\varepsilon_\tau)\ell^2$, where $\Sigma_1$ is the source
covariance of Lemma~\ref{lem:source-cov}. By the isotropy
estimate~\eqref{eq:sigma-iso}, $\Sigma_1=c_d\ell^2 I+R_1$ with
$\|R_1\|=O(\delta_{\mathrm{cov}})\ell^2$, and the Gaussian reference
moment fixes the dimensional constant exactly to $c_d=\tfrac12$. Substituting,
\begin{equation}\label{eq:DPhi-clean}
   D\Phi=\tfrac{2}{\ell^2}\bigl[\hat\Lambda(\tfrac12\ell^2 I+R_1)+E_{\mathrm{cr}}\bigr]
   +O(\ell^2/\lambda^2)
   =\hat\Lambda+E_\Phi,\qquad \|E_\Phi\|=O(\theta),
\end{equation}
with $\theta:=\varepsilon_\tau+\delta_{\mathrm{cov}}+\ell^2/\lambda^2$.
The normalization $2/\ell^2$ (from the gradient of the Gaussian weight) and the moment
$c_d=\tfrac12$ multiply to exactly $1$, so the Gaussian--Karcher differential is
automatically scale-normalized and no conformal factor survives ($\Omega_0:=2c_d=1$).
Moreover $\hat\Lambda$ is an exact Lorentz transformation: the Procrustes projection
(Lemma~\ref{lem:finite-procrustes}) produces it from anchors whose Lorentzian
separations are preserved by~(F3) at the common density $\rho$, so the scale is pinned
to unity rather than merely up to a dilation.

Since $\hat\Lambda^T\eta\hat\Lambda=\eta$ exactly, the pullback follows directly.
Working in $h_1$- and $h_2$-normal coordinates at $x$ and $\Phi(x)$ (where
$g_1,g_2=\eta+O(\ell^2/\lambda^2)$),
\[
   \Phi^*g_2=(D\Phi)^T\eta\,(D\Phi)+O(\ell^2/\lambda^2)
   =(\hat\Lambda+E_\Phi)^T\eta(\hat\Lambda+E_\Phi)+O(\ell^2/\lambda^2)
   =\eta+O(\|E_\Phi\|)+O(\ell^2/\lambda^2),
\]
that is, $\Phi^*g_2=g_1+E$ with $\|E\|_{g_1}=O(\theta)$.
\end{proof}

\begin{remark}[The scale is pinned by F3, not by a separate volume step]
\label{rem:scale-by-f3}
In the Malament dichotomy, causal order fixes the metric only up to a local conformal
factor $\Omega$, which volume information then pins. Here the embedding preserves
proper times~(F3), not merely causal order, and proper time carries the scale; this is
why the Procrustes map already lies in $O(1,d-1)$ rather than in the conformal group,
equivalently why $\Omega_0=2c_d=1$ exactly in~\eqref{eq:DPhi-clean}. No separate
volume-rigidity step is needed to fix $\Omega$. The volume condition~(F2) still enters,
but only through the isotropy and non-degeneracy of $\Sigma_1$
(Lemma~\ref{lem:source-cov}) and the existence of the scaffold (Lemma~\ref{lem:scaffold}).
\end{remark}

\subsection{Assembling the main theorem}

\begin{theorem}[Approximate isometry of well-conditioned embeddings]
\label{thm:haupt}
Let $C$ be a finite causal set with two $(\rho;c_*,K_d)$-well-conditioned embeddings $f_1\colon C\hookrightarrow(M_1,g_1)$ and $f_2\colon C\hookrightarrow(M_2,g_2)$ into globally hyperbolic $d$-dimensional Lorentzian manifolds with curvature scale $\lambda$, satisfying $\rho\lambda^d\geq c_*^{-2d}(\log\rho V_{\max})^2$.

Then on the open subdomains $M_1^\circ,M_2^\circ$ (Equation~\eqref{eq:M1circ-def}) there exists a smooth diffeomorphism $\Phi\colon M_1^\circ\to M_2^\circ$ satisfying $\Phi^*g_2=g_1+E$ with
\begin{equation}\label{eq:eps-final}
   \|E\|_{g_1}
   = O\bigl(\rho^{-2/(5d)}\lambda^{-2/5}\,\log^{3/2}(\rho V_{\max})\bigr),
\end{equation}
where $V_{\max}:=\max(\mathrm{Vol}_{g_1}(M_1),\mathrm{Vol}_{g_2}(M_2))$.\footnote{The power law $\rho^{-2/(5d)}\lambda^{-2/5}$ is sharp for this argument; the logarithmic exponent $3/2$ is a valid upper bound, a careful accounting of the balance at $\ell_*$ giving $\log^{6/5}$.}
\end{theorem}

\begin{proof}[Proof of Theorem~\ref{thm:haupt}]
We assemble the results of Part~I. The moving Karcher mean $\Phi$
(Construction~\ref{con:phi}) is well-defined on $M_1^\circ$
(Proposition~\ref{prop:well-def}), the localization of its active target points
(Corollary~\ref{cor:target-localization}) resting only on the timelike
trilateration and so independent of $\Phi$. By Proposition~\ref{prop:basic} it
is smooth with $\|D\Phi\|=O(1)$, and by Proposition~\ref{prop:local-diffeo} it
is a local diffeomorphism, since
$D\Phi=(2/\ell^2)\Sigma_{\mathrm{cross}}+O(\ell^2/\lambda^2)$ with
$\Sigma_{\mathrm{cross}}$ non-degenerate (Lemma~\ref{lem:cross-cov}). The degree
argument of Proposition~\ref{prop:diffeo} upgrades this to a global
diffeomorphism $\Phi\colon M_1^\circ\to M_2^\circ$. Distance preservation~(F3),
at the common density $\rho$, makes the Procrustes map an exact Lorentz
transformation, so by Proposition~\ref{prop:conformal} $\Phi$ is an approximate
isometry, $\Phi^*g_2=g_1+E$ with
$\|E\|_{g_1}=O(\varepsilon_\tau+\ell^2/\lambda^2+\delta_{\mathrm{cov}})$;
the scale requires no separate volume-rigidity step (Remark~\ref{rem:scale-by-f3}).
Optimizing the smoothing scale $\ell$ balances the curvature term $\ell^2/\lambda^2$
against the Bollob\'as--Brightwell term in $\varepsilon_\tau$ at
$\ell_*=\Theta(\lambda^{4/5}\rho^{-1/(5d)})$, where the covariance term
$\delta_{\mathrm{cov}}=\Theta(\rho^{-4/(5(d+2))}\cdots)$ is subdominant for $d>2$;
this yields~\eqref{eq:eps-final}.
\end{proof}

\begin{remark}[Structure of the bound]
The error~\eqref{eq:eps-final} is of a single origin, the conformal-rigidity term
\[
   O\bigl((\rho\ell^d)^{-\alpha_d}+\ell^2/\lambda^2\bigr),
\]
optimized in the smoothing scale $\ell$ using the Bollob\'as--Brightwell rate
$\alpha_d=1/(2d)$ from~\cite{BollobasBrightwell1992} (Theorems~3 and~9), with optimum
at $\ell_*=\rho^{-1/(5d)}\lambda^{4/5}$. The scale is fixed directly by~(F3), so there
is no separate volume-rigidity contribution (Remark~\ref{rem:scale-by-f3}).
\end{remark}

\section{Part II: A Poisson sprinkling is Well-Conditioned (a.s.)}

\paragraph{Roadmap} We now show a Poisson sprinkling is well-conditioned almost surely in the high-density limit. Order-preservation (F1) is automatic. Volume-faithfulness (F2) follows from Chernoff concentration with a union bound over a mesoscopic net. The longest-chain/proper-time correspondence (F3) is the Bollob\'as-Brightwell theorem transferred to the curved setting by a cone-sandwiching argument. The scaffold condition is not separate: it follows deterministically from volume-faithfulness (Lemma~\ref{lem:scaffold}), hence holds once (F2) does. Each holds with probability $1-(\rho V_M)^{-K'}$; a union bound gives all simultaneously.  

\begin{theorem}[Poisson sprinklings are almost surely well-conditioned]
\label{thm:poisson-wc}
Let $(M,g)$ be a globally hyperbolic $d$-dimensional Lorentzian manifold with finite
$g$-volume $V_M$ and curvature scale $\lambda$, and fix $\rho>0$ and $c_*\in(0,1)$ with
$\rho\lambda^d\geq c_*^{-2d}(\log\rho V_M)^2$. Let $C:=\Pi$ be a Poisson process of
intensity $\rho\,dV_g$, ordered by $J^-_g$. Then for every sufficiently large dimensional
constant $K_d$, the inclusion $f\colon C\hookrightarrow M$ is a
$(\rho;c_*,K_d)$-well-conditioned embedding (Definition~\ref{def:faithful}) with
probability at least
\begin{equation}\label{eq:poisson-wc-prob}
   1-(\rho V_M)^{-K_d'},
\end{equation}
where $K_d'>0$ is a dimensional constant growing with $K_d$.
\end{theorem}

\begin{proof}
We verify the three conditions (F1)--(F3) and take a union bound. Order-preservation~(F1)
and scale-dependent density~(F2) hold with probability $\geq 1-(\rho V_M)^{-K_d'}$ by
Lemma~\ref{lem:poisson-faithful}. Distance preservation~(F3) requires the single-pair estimate of
Proposition~\ref{prop:longest-chain} (the Bollob\'as--Brightwell longest-chain
estimate transferred to $(M,g)$ there by the curvature-sandwich and
bounded-differences argument) to hold simultaneously over the $O((\rho V_M)^2)$
admissible pairs of the mesoscopic net used for~(F2). Proposition~\ref{prop:longest-chain}
gives each pair a failure probability $\leq(\rho V_M)^{-K'+o(1)}$ with $K'$ free; choosing
$K'>2+K_d'$ and taking a union bound over the pairs leaves total failure
$\leq(\rho V_M)^{-K_d'}$, so (F3) holds with probability $\geq 1-(\rho V_M)^{-K_d'}$.
The anchor scaffold is not a separate condition: it holds deterministically given (F2), by
Lemma~\ref{lem:scaffold}, so no separate probabilistic estimate is required. Taking the
intersection of the two high-probability events and bounding the union of their complements
(each failure event has probability at most $(\rho V_M)^{-K_d'}$, the factor $2$ absorbed into
a slightly smaller $K_d'$), all of (F1)--(F3) hold simultaneously with probability at
least~\eqref{eq:poisson-wc-prob}.
\end{proof}
We proceed by showing that Poisson embeddings are faithful embeddings, that is, they follow $(F1)$--$(F2)$, and then that they obey distance preservation $(F3)$; the anchor scaffold is then automatic by Lemma~\ref{lem:scaffold}.

\begin{lemma}[Poisson sprinklings are faithful embeddings]
\label{lem:poisson-faithful}
Let $(M,g)$ be a globally hyperbolic Lorentzian manifold of dimension $d$ with finite $g$-volume $V_M$ and curvature scale $\lambda$. Fix $\rho>0$ and a dimensional constant $c_*\in(0,1)$ such that $\rho\lambda^d\geq c_*^{-2d}(\log\rho V_M)^2$. Let $\Pi\subset M$ be a Poisson process of intensity $\rho\cdot dV_g$, and let $C:=\Pi$ with the partial order induced by $J^-_g$. Then for any constant $K_d>0$ sufficiently large, the inclusion $f:C\hookrightarrow M$ is a $(\rho;c_*,K_d)$-faithful embedding with probability at least
\begin{equation}\label{eq:poisson-prob}
   1-(\rho V_M)^{-K_d'},
\end{equation}
where $K_d'>0$ is a dimensional constant growing with $K_d$.
\end{lemma}

\begin{proof}
\emph{(F1) Order-preservation.} By construction: the induced partial order on $\Pi$ is the restriction of $J^-_g$, so $x\preceq y\iff f(x)\in J^-_g(f(y))$ holds exactly.

\emph{(F2) Scale-dependent uniform density.} The admissible range is $I:=[\tau_{\min},c_*\lambda]$ with $\tau_{\min}:=c_*^{-1}\rho^{-1/d}(\log\rho V_M)^{2/d}$. For each admissible diamond $D$, the count $N_D:=|\Pi\cap D|$ is Poisson with mean $\mu_D:=\rho\,\mathrm{Vol}_g(D)$. By the Chernoff bound, for any $t\in(0,1)$,
\[
   \Pr(|N_D-\mu_D|\geq t\mu_D)\leq 2\exp(-c\,t^2\mu_D).
\]
\emph{Setting $t=\delta_D$ with $\delta_D^2\mu_D=K_d^2\log\rho V_M$ makes the exponent independent of $D$:}
\[
   \delta_D^2\mu_D=K_d^2\log(\rho V_M)
   \quad\implies\quad
   \delta_D=K_d\sqrt{\frac{\log\rho V_M}{\mu_D}} =K_d\sqrt{\frac{\log\rho V_M}{\rho\,\mathrm{Vol}_g(D)}}.
\]
This is exactly the scale-dependent tolerance in~\eqref{eq:F2-tolerance}. The per-diamond failure probability is thus $2\exp(-c K_d^2\log\rho V_M)=2(\rho V_M)^{-cK_d^2}$, uniformly over all admissible diamonds.

\emph{Union bound.}
The admissible parameter space is $\{(p,q)\in M\times M:\tau_g(p,q)\in I, d_h(D,\partial M)\geq c_*\lambda\}\subset M\times M$, parametrized by two endpoints in the $d$-dimensional manifold $M$. Covering $M$ by $h$-balls of radius $\tau_{\min}$ requires at most $V_M/\tau_{\min}^d\cdot(1+o(1))$ balls (standard volume estimate on Riemannian manifolds at sub-curvature scales). Hence the product space $M\times M$ has covering number $\leq(V_M/\tau_{\min}^d)^2(1+o(1))$ on a $\tau_{\min}$-net. With $\tau_{\min}=c_*^{-1}\rho^{-1/d}(\log\rho V_M)^{2/d}$,
\[
   \frac{V_M}{\tau_{\min}^d}
   =\frac{V_M}{c_*^{-d}\rho^{-1}(\log\rho V_M)^2}
   =\frac{c_*^d\,\rho V_M}{(\log\rho V_M)^2},
\]
so $N_*\leq C(\rho V_M)^2/(\log\rho V_M)^4$. The dimension $d$ cancels because the per-endpoint count $V_M/\tau_{\min}^d=O(\rho V_M)$ exactly absorbs the $d$-dependent volume of small balls. Hence $N_*\leq C(\rho V_M)^2$ uniformly in $d$, as claimed.

F2 holds at all net points simultaneously with probability
\[
   \geq 1-2N_*(\rho V_M)^{-cK_d^2} \geq 1-(\rho V_M)^{2-cK_d^2}.
\]
Given any target probability exponent $K_d'>0$, choose $K_d\geq\sqrt{(K_d'+2)/c}$ so that $cK_d^2-2\geq K_d'$, yielding the failure probability bound $(\rho V_M)^{-K_d'}$ stated in the lemma. The constants $K_d$ and $K_d'$ thus play distinct roles: $K_d$ is the multiplicative constant in the tolerance $\delta_D=K_d\sqrt{\log\rho V_M/(\rho\,\mathrm{Vol}_g(D))}$, while $K_d'$ is the exponent governing the failure probability $(\rho V_M)^{-K_d'}$. We may freely enlarge both, with $K_d'$ depending quadratically on $K_d$.

\emph{Continuity between net points and rate-scaling check.} The above bound holds at all net points; extending to all admissible diamonds requires a continuity argument controlling the variation of $\mathrm{Vol}_g(D)$ under endpoint perturbation by $\tau_{\min}$. The key estimate is $|\partial_p\mathrm{Vol}_g(D)|\leq C\tau^{d-1}$ via the divergence theorem at sub-curvature scales. Combined with the fact that the F2 tolerance $\delta_D\cdot\mathrm{Vol}_g(D) \sim K_d\sqrt{\log\rho V_M\cdot\mathrm{Vol}_g(D)/\rho}$ strictly dominates the perturbation, F2 transfers to all admissible diamonds with constants absorbed into $K_d$. At the relevant scales (from $\tau_{\min}$ up through the smoothing scale $\ell$), the tolerance $\delta_D$ matches the polynomial convergence rate of the main theorem. Full calculation in Appendix~\ref{app:poisson-faithful-details}.
\end{proof}

\paragraph{Poisson longest chain approximates proper time.}

\begin{proposition}[Mesoscopic longest chain of Poisson sprinkling]
\label{prop:longest-chain}
Let $x\prec y$ in $C$ with $f_i(x),f_i(y)$ contained in a mesoscopic diamond of diameter $L$ and Minkowski volume $V\sim L^d$ in $(M_i,g_i)$, where $i\in\{1,2\}$. Fix any tail exponent $K'>0$. Then with probability $\geq 1-(\rho V)^{-K'+o(1)}$:
\begin{equation}\label{eq:chain-time}
   \biggl|\frac{\ell_C(x,y)}{(m_d\rho)^{1/d}}
   - \tau_i(f_i(x),f_i(y))\biggr|
   \;\leq\;
   C_d\Bigl(\frac{L^2}{\lambda^2} +\frac{\log^{3/2}(\rho V)}{(\rho V)^{1/(2d)}}\Bigr) \cdot\tau_i,
\end{equation}
where $m_d$ is the Myrheim--Meyer constant for the Minkowski causal order, $\tau_i$ is the proper time in $(M_i,g_i)$, and the convergence exponent $1/(2d)$ is that of Theorems~3 and~9 of Bollob\'as--Brightwell~\cite{BollobasBrightwell1992}.
\end{proposition}

\begin{proof}
The proof reduces the curved-space problem to flat-space combinatorics via a geometric sandwiching argument.

\textit{Step~1: Local flat approximation.}
Choose a base point midway between $p:=f_i(x)$ and $q:=f_i(y)$ and introduce Riemann normal coordinates. In these coordinates: $g_{\mu\nu}=\eta_{\mu\nu}-\frac{1}{3}R_{\mu\alpha\nu\beta}z^\alpha z^\beta+O(|z|^3/\lambda^3)$. Since the entire diamond has coordinate radius $\leq L$, the metric perturbation $E:=g-\eta$ satisfies, for all vectors $v=(v^0,\vec{v})$:
\begin{equation}\label{eq:metric-perturb}
   |E(v,v)| \;\leq\; \varepsilon\bigl((v^0)^2+|\vec{v}|^2\bigr),
   \qquad \varepsilon:=C_d'L^2/\lambda^2.
\end{equation}
Equivalently: $-(1{+}\varepsilon)(v^0)^2+(1{-}\varepsilon)|\vec{v}|^2 \leq g(v,v) \leq -(1{-}\varepsilon)(v^0)^2+(1{+}\varepsilon)|\vec{v}|^2$.

\textit{Step~2: Sandwiching the causal diamond.} Define two auxiliary Minkowski metrics with modified light-cone angles:
\[
   \eta^+(v,v) := -\frac{1{+}\varepsilon}{1{-}\varepsilon}(v^0)^2+|\vec{v}|^2,
   \qquad
   \eta^-(v,v) := -\frac{1{-}\varepsilon}{1{+}\varepsilon}(v^0)^2+|\vec{v}|^2.
\]
If $v$ is $g$-causal ($g(v,v)\leq 0$), the lower bound gives $(1{-}\varepsilon)|\vec{v}|^2\leq(1{+}\varepsilon)(v^0)^2$, hence $\eta^+(v,v)=-\frac{1{+}\varepsilon}{1{-}\varepsilon}(v^0)^2 +|\vec{v}|^2\leq 0$. Conversely, if $v$ is $\eta^-$-causal ($\eta^-(v,v)\leq 0$, i.e., $|\vec{v}|^2\leq\frac{1{-}\varepsilon}{1{+}\varepsilon}(v^0)^2$), then $g(v,v)\leq -(1{-}\varepsilon)(v^0)^2 +(1{+}\varepsilon)\cdot\frac{1{-}\varepsilon}{1{+}\varepsilon}(v^0)^2=0$. Therefore:
\[
   D_{\eta^-}(p,q) \;\subseteq\; D_g(p,q) \;\subseteq\; D_{\eta^+}(p,q).
\]
Since $\eta^\pm$ are related to $\eta$ by a rescaling of the time coordinate $t'=\sqrt{(1{\pm}\varepsilon)/(1{\mp}\varepsilon)}\,t =(1\pm\varepsilon+O(\varepsilon^2))t$, the proper times satisfy $\tau_\pm=\tau_g(1\pm O(L^2/\lambda^2))$.

\textit{Step~3: Monotonicity of the longest chain.} The \emph{same} Poisson point set $\{p_k\}$ is used for all three metrics; only the causal order changes. The nesting $D^-\subseteq D_g\subseteq D^+$ implies that every chain valid in $(\{p_k\},\preceq_{\eta^-})$ is also valid in $(\{p_k\},\preceq_g)$, and every chain valid in $(\{p_k\},\preceq_g)$ is valid in $(\{p_k\},\preceq_{\eta^+})$. Hence
\begin{equation}\label{eq:chain-squeeze}
   \ell_{C_{\eta^-}}(x,y) \;\leq\; \ell_{C_g}(x,y) \;\leq\; \ell_{C_{\eta^+}}(x,y).
\end{equation}
This is a \emph{deterministic} inequality for each realization of the Poisson process.

\textit{Step~4: Flat-space longest-chain estimate.} The bounding sets $C_{\eta^\pm}$ are Poisson sprinklings at density $\rho$ into Minkowski causal diamonds, with expected count $n=\rho V_\pm$. Two facts about the Minkowski longest chain $H$ are needed: the value of its mean and the size of its fluctuations.

\emph{Mean.} The longest-chain statistics of a Poisson sprinkling of a Minkowski causal diamond are classical in the causal set literature~\cite{BrightwellGregory1991}: the expected longest chain between the tips of a diamond of proper-time height $\tau$ satisfies
\begin{equation}\label{eq:mink-mean}
   \mathbb{E}H = (m_d\rho)^{1/d}\,\tau\,(1+o(1)),
\end{equation}
where $m_d>0$ is the Myrheim--Meyer constant.

\emph{Fluctuations.} The relative deviation of $H$ from its mean is $O(n^{-1/(2d)}\log^{3/2}n)$. This is the content of the Bollob\'as--Brightwell estimates, which we transfer to the Minkowski order through their proof rather than through the orders themselves. Bollob\'as and Brightwell~\cite{BollobasBrightwell1992} prove their estimates for the coordinatewise order on $[0,1]^d$, but the proof uses only two ingredients: the method of bounded differences (their Lemma~4, a McDiarmid-type concentration inequality for functions of independent variables, itself proved via Azuma's martingale bound) and a partition of the domain into independent strips transverse to the order direction such that any chain meets few points per strip. Both ingredients hold verbatim for the Minkowski causal order. Slicing the diamond by the time coordinate into slabs of equal thickness yields independent Poisson strips, and causality confines any causal chain within a slab of time-thickness $\delta$ to a spatial ball of radius $\leq\delta$ (each causal step satisfies $|\Delta\vec{x}|\leq\Delta x^0$), hence to $O(1)$ cells of the spatial subdivision; this supplies the bounded-difference constant, while the Poisson tail bounding the surplus of over-full cells is unchanged. The argument therefore yields, with $d$-dependent constants, the same two estimates Bollob\'as and Brightwell obtain in the coordinate model.

We record those estimates in the coordinate-model notation in which the explicit constants are stated. Let $H_{n,d}$ denote the longest chain in the coordinate-order Poisson model on $[0,1]^d$, write $\mathbb{E}H_{n,d}=c_{n,d}\,n^{1/d}$, and let $c_d=\lim_{n}c_{n,d}$.

\emph{Theorem~9 (rate of mean convergence):}
\begin{equation}\label{eq:bb-mean}
   c_d \;\geq\; c_{n,d} \;\geq\; c_d - \frac{12 K_d^{\mathrm{BB}} \log^{3/2}n}{n^{1/(2d)}\log\log n},
\end{equation}
hence $|\mathbb{E}H_{n,d}-c_d n^{1/d}| \leq 12 K_d^{\mathrm{BB}}\,n^{1/(2d)}\log^{3/2}n/\log\log n$.

\emph{Theorem~3 (concentration):} For $n$ sufficiently large and $2<u<n^{1/(2d)}/\log\log n$,
\begin{equation}\label{eq:bb-conc}
   \Pr\!\Bigl(|H_{n,d}-\mathbb{E}H_{n,d}| > \frac{u\, K_d^{\mathrm{BB}}\,n^{1/(2d)}\log n}{\log\log n}\Bigr) \;\leq\; 4u^2 e^{-u^2}.
\end{equation}
Here $u$ is the Bollob\'as--Brightwell deviation parameter (unrelated to the curvature scale $\lambda$), and $K_d^{\mathrm{BB}}$ is their dimensional constant (unrelated to the F2 tolerance constant $K_d$ of Definition~\ref{def:faithful}).

Fix any tail exponent $K'>0$ and set $u=\sqrt{K'\log n}$ in~\eqref{eq:bb-conc}: the deviation is bounded by $\sqrt{K'}\,K_d^{\mathrm{BB}}\,n^{1/(2d)}\log^{3/2}n/\log\log n$ with probability $\geq 1-4K'\log n\cdot n^{-K'}$. Combining with~\eqref{eq:bb-mean} via the triangle inequality gives $|H_{n,d}-c_d n^{1/d}|\leq C_d\,n^{1/(2d)}\log^{3/2}n$ (the factor $\sqrt{K'}$ absorbed into $C_d$) with probability $\geq 1-n^{-K'+o(1)}$, that is, $H_{n,d}=\mathbb{E}H_{n,d}\,(1+O(n^{-1/(2d)}\log^{3/2}n))$. This relative bound is normalization-free, so by the transfer above it holds equally for the Minkowski longest chain about its mean~\eqref{eq:mink-mean}. Substituting $n=\rho V_\pm$ and dividing~\eqref{eq:mink-mean} by $(m_d\rho)^{1/d}$:
\[
   \biggl|\frac{\ell_{C_{\eta^\pm}}(x,y)}{(m_d\rho)^{1/d}} -\tau_\pm\biggr|
   \;\leq\; C_d'\,\frac{\log^{3/2}(\rho V_\pm)}{(\rho V_\pm)^{1/(2d)}} \,\tau_\pm.
\]

\textit{Step~5: Synthesis.}
With probability $\geq 1-(\rho V)^{-K'+o(1)}$ (from the tail choice above), the BB bounds hold simultaneously for both bounding diamonds. Substituting $\tau_\pm=\tau_g(1\pm O(L^2/\lambda^2))$ and applying the squeeze~\eqref{eq:chain-squeeze}:
\[
   \frac{\ell_{C_g}}{(m_d\rho)^{1/d}} \;\leq\; \tau_g\bigl(1+O(L^2/\lambda^2)\bigr) + C'\tau_g\,\frac{\log^{3/2}(\rho V)}{(\rho V)^{1/(2d)}},
\]
and similarly for the lower bound. Consolidating the error terms gives~\eqref{eq:chain-time}.
\end{proof}

\begin{remark}[Extension to non-compact spacetimes]
\label{rem:noncompact}
For a non-compact globally hyperbolic spacetime $M$ (e.g.,
Minkowski, FLRW, de~Sitter), the conclusion holds in the
following regional form. Fix a precompact region
$K\subset M$ of finite volume $V_K$, restrict the causal set
to $C_K:=f^{-1}(K)$, and apply Theorem~\ref{thm:haupt} with
$V_M$ replaced by $V_K$. Then with probability
$\geq 1-(\rho V_K)^{-K_d'}$ over the sprinkling there is a
diffeomorphism $\Phi_K\colon K^\circ\to\Phi(K^\circ)$ with
error $O(\rho^{-2/(5d)}\lambda^{-2/5}\log^{3/2}(\rho V_K))$;
as $\rho\to\infty$ with $K$ fixed, both the error and the
failure probability tend to zero. The statement is thus a
family of high-probability local results, one for each
precompact $K$.

A single almost-sure statement over all of $M$ is not
available when $\mathrm{Vol}_g(M)=\infty$. The
well-conditioning conditions (F1)--(F3) are established by a
union bound over a net in $M$, each cell of which fails with
positive probability; an infinite-volume manifold contains
infinitely many independent cells, so by the second
Borel--Cantelli lemma the event that infinitely many cells
fail has probability one (see \cite{Borel1909,Cantelli1917}).
For physical applications
one fixes a reference volume $V_K$ (e.g.\ a Hubble volume)
and works within the corresponding region.
\end{remark}

\begin{remark}[Boundary layer vanishes in the limit]
\label{rem:boundary-shrinks}
The interior subdomain $M^\circ$ (defined in
\eqref{eq:M1circ-def}) excludes an $h$-neighborhood of
$\partial M$ of width $c_*\lambda+12\alpha\ell$. The
mesoscopic part ($12\alpha\ell$) vanishes as
$\rho\to\infty$, but the macroscopic part ($c_*\lambda$)
does not. It is the boundary clearance demanded by
the F2 condition (Definition~\ref{def:faithful}), which
excludes diamonds within $c_*\lambda$ of $\partial M$
because the embedded count there would be truncated by the
boundary. We therefore distinguish two cases:

\emph{(a) No boundary or exhaustion.} For spatially compact
spacetimes (Cauchy slices, no spatial boundary) or for
spacetimes treated by precompact exhaustion
$K_n\nearrow M$, one applies the theorem on each $K_n$ and
the excluded layer is a boundary effect that recedes as the
exhaustion proceeds. In these cases
$\mathrm{Vol}_g(M\setminus M^\circ)/\mathrm{Vol}_g(M)\to 0$.

\emph{(b) Fixed manifold with boundary.} For a fixed
spacetime-with-boundary and fixed $\lambda$, the excluded
layer of width $c_*\lambda$ does not shrink as
$\rho\to\infty$; the theorem establishes the approximate
isometry only on the deep interior
$\{d_{h_1}(\cdot,\partial M)>c_*\lambda\}$. This is the
unavoidable consequence of the geometric spread of causal
diamonds: a mesoscopic diamond centered within $c_*\lambda$
of the boundary is geometrically truncated by $\partial M$.
Since the volume-rigidity argument requires exact, complete
causal diamonds, the F2 uniform density bound cannot be
applied within this boundary layer.
\end{remark}

\begin{corollary}[Approximate isometry of Poisson sprinklings]
\label{cor:quantitative}
Let $C$ be a causal set presented as a Poisson sprinkling at common density $\rho$ of
each of two globally hyperbolic $d$-dimensional spacetimes $(M_1,g_1)$ and $(M_2,g_2)$ with
curvature scale $\lambda$, via inclusions $f_1,f_2$, with
$\rho\lambda^d\geq c_*^{-2d}(\log\rho V_{\max})^2$ and
$V_{\max}:=\max(\Vol_{g_1}(M_1),\Vol_{g_2}(M_2))$. Then with probability at least
$1-2(\rho V_{\max})^{-K_d'}$ there is a smooth diffeomorphism
$\Phi\colon M_1^\circ\to M_2^\circ$ with $\Phi^*g_2=g_1+E$ and
\begin{equation}\label{eq:eps-poisson}
   \|E\|_{g_1}
   \;=\;O\Bigl(\rho^{-2/(5d)}\lambda^{-2/5}\log^{3/2}(\rho V_{\max})\Bigr)
   \;=:\;\varepsilon(\rho,\lambda).
\end{equation}
Both $\varepsilon(\rho,\lambda)\to 0$ and the failure probability $\to 0$ as
$\rho\lambda^d\to\infty$. In particular, a single causal set cannot be a faithful Poisson
sprinkling of two macroscopically distinct spacetimes: in the high-density limit the two
geometries are forced to agree, up to the explicit error~\eqref{eq:eps-poisson}, on the
deep interiors $M_1^\circ,M_2^\circ$.
\end{corollary}

\begin{proof}
The hypotheses are two probability statements about the single combinatorial object $C$,
coupled only through $C$; we do not require them to be independent. Applying
Theorem~\ref{thm:poisson-wc} to $f_1$ and to $f_2$ separately, each presents $C$ as a
$(\rho;c_*,K_d)$-well-conditioned embedding with probability $\geq 1-(\rho V_{\max})^{-K_d'}$.
By a union bound both hold simultaneously with probability $\geq 1-2(\rho V_{\max})^{-K_d'}$.
On that event the hypotheses of Theorem~\ref{thm:haupt} are met, and its conclusion supplies
$\Phi$, $E$ together with the rate~\eqref{eq:eps-final}, which
is~\eqref{eq:eps-poisson}.
\end{proof}

\section{Discussion}

\subsection{Proof architecture: a variant of ``order plus number''}
\label{subsec:architecture}

The causal set program rests on the slogan that order determines the conformal (causal) structure while number determines the scale: Malament's theorem~\cite{Malament1977} supplies the first half, and the volume-counting of Bombelli--Lee--Meyer--Sorkin~\cite{Bombelli1987} the second. Our proof follows this division only in part, and it is worth being explicit about where it departs.

In our argument the causal order does double duty. The cone-preservation that drives the local Procrustes alignment (Lemma~\ref{lem:finite-procrustes}) is the conformal half, as expected. But the metric scale is also fixed by the order, through the longest-chain/proper-time correspondence (F3, Proposition~\ref{prop:longest-chain}): the longest chain between two elements is an invariant of the abstract poset, so at common density it pins the proper time in both spacetimes at once. We therefore do not use volume information to set the scale, and the volume-rigidity step of the traditional account is absent (Remark~\ref{rem:scale-by-f3}).

The volume condition (F2) instead plays a different structural role: it forces the local point cloud to be isotropic and non-degenerate (Lemma~\ref{lem:source-cov}), which is what makes the moving Karcher mean a local diffeomorphism (Proposition~\ref{prop:local-diffeo}) and, after a degree argument, a global one (Proposition~\ref{prop:diffeo}). In the slogan ``order plus number,'' number has thus moved from setting the scale to guaranteeing the
reconstruction is non-degenerate; the scale is carried entirely by the order.

\subsection{Relation to prior work}

Malament~\cite{Malament1977} proves exact uniqueness from the full causal structure: a continuous, time-orientation-preserving bijection of $d\geq 3$-dimensional globally hyperbolic spacetimes preserving causal order is necessarily a smooth conformal isometry. Together with volume information this would determine the full metric, and Malament's result is the conceptual origin of the slogan ``order plus number gives geometry.''

M\"uller~\cite{Mueller2025} examined what happens when one tries to formalize this slogan for finite causal sets. He considered two versions:

(i) A direct \emph{finite-set} version, where a finite poset is faithfully embedded into two spacetimes by maps preserving only the causal order. M\"uller showed this version is false: he constructs examples of finite causal sets admitting causal-order-preserving embeddings into geometrically distinct manifolds. The underlying issue is that without volume information, finite causal data is too sparse to pin down the manifold; an adversary can rearrange the embedded points in different spacetimes while preserving the order.

(ii) A \emph{countable-set} version with stronger hypotheses: the embedded points form a dense subset of the manifold, and the embedding is required to satisfy a ``Planck-scale uniform density'' condition that is deterministic rather than Poisson. Under these stronger hypotheses, M\"uller proves an exact isometry result. However, the deterministic uniform-density condition does not correspond to Poisson sprinkling, and the countable limit removes the finite-density structure that is essential to causal set theory.

Our result complements M\"uller's countable-set theorem by
extending the uniqueness framework to finite Poisson
sprinklings: we work with finite causal sets, with volume
information~(F2) encoding the statistical Poisson structure
rather than a deterministic Planck-scale condition, and obtain
quantitative approximate uniqueness with an explicit
high-density limit. As discussed in
\S\ref{subsec:architecture}, F2 enters our proof through the
non-degeneracy of the local reconstruction rather than through
a volume-rigidity step, while the scale is carried by the
longest-chain correspondence; our error bounds make precise
what ``approximate'' means as $\rho\to\infty$.

The Brightwell--Gregory~\cite{BrightwellGregory1991} longest chain estimator, originally introduced as a dimension estimator, here plays a structural role: it provides the combinatorial invariant that bridges the two embeddings.

The moving Karcher mean construction may be of independent interest. It simultaneously serves as a smooth, global interpolation of noisy point correspondences on Riemannian manifolds, defined by a single intrinsic minimization problem, and as a statistical low-pass filter that suppresses high-frequency noise in the derivative through the delta-method analysis of weighted ratios.

\section{Acknowledgments}
The author thanks David Morrison for valuable discussions and guidance throughout the development of this work.

\appendix

\section{Algebra of finite Lorentzian Procrustes}
\label{app:procrustes-algebra}

We give the full algebraic proof of
Lemma~\ref{lem:finite-procrustes} (finite Lorentzian
Procrustes).

\begin{proof}
Identify each $a_m,b_m\in\R^d$ as a column vector, and
form the $d\times d$ matrices
\[
   A:=[a_1\,a_2\,\cdots\,a_d],
   \qquad
   B:=[b_1\,b_2\,\cdots\,b_d]
\]
whose $m$-th columns are $a_m,b_m$ respectively. By (C3),
$A$ is invertible with $\|A^{-1}\|=O(1/(\alpha\ell))$.
Define the $\eta$-Gram matrices
\[
   G^a:=A^T\eta A,
   \qquad
   G^b:=B^T\eta B,
\]
so $G^a_{ij}=a_i^T\eta a_j=\eta(a_i,a_j)$ and analogously
for $G^b$.

\emph{Step 1 (Gram matrix preservation).}
By polarization,
$2\eta(a_i,a_j)=\eta(a_i,a_i)+\eta(a_j,a_j)-\eta(a_i-a_j,a_i-a_j)$
and analogously for $b$. Applying~\eqref{eq:proc-hyp} to
the three squared distances yields
$|\eta(a_i,a_j)-\eta(b_i,b_j)|\leq\tfrac{3}{2}\varepsilon$
for all $1\leq i,j\leq d$, so
\[
   \|G^a-G^b\|\leq\tfrac{3d}{2}\varepsilon.
\]

\emph{Step 2 (defect of the candidate Lorentz transformation).}
Define the candidate Lorentz map $\Lambda:=BA^{-1}$. Then
$B=\Lambda A$, equivalently $b_m=\Lambda a_m$ for each $m$.
Compute
\[
   \Lambda^T\eta\Lambda
   = (BA^{-1})^T\eta(BA^{-1})
   = A^{-T}B^T\eta B A^{-1}
   = A^{-T}G^b A^{-1}.
\]
For an exact Lorentz transformation, $\Lambda^T\eta\Lambda=\eta$;
this would require $G^b=A^T\eta A=G^a$. The deviation is
\[
   \Lambda^T\eta\Lambda - \eta
   = A^{-T}(G^b-G^a)A^{-1},
\]
so
\begin{equation}\label{eq:Lambda-defect}
   \|\Lambda^T\eta\Lambda-\eta\|
   \leq \|A^{-1}\|^2\cdot\|G^b-G^a\|
   = O\bigl(\varepsilon/(\alpha\ell)^2\bigr).
\end{equation}

\emph{Step 3 (projection to the Lorentz group).}
We apply Lemma~\ref{lem:approx-az} to $\Lambda$. That
lemma's hypothesis is approximate null-cone preservation:
$|\eta(\Lambda v,\Lambda v)|\leq\theta'|\Lambda v|^2$ for
all null $v$. Bounding the left side via~\eqref{eq:Lambda-defect}:
\[
   |\eta(\Lambda v,\Lambda v)|
   = |v^T\Lambda^T\eta\Lambda\,v|
   = \Bigl|\,v^T(\Lambda^T\eta\Lambda-\eta)v
     +\underbrace{v^T\eta v}_{=\,0}\,\Bigr|
   \leq \|\Lambda^T\eta\Lambda-\eta\|\cdot|v|^2,
\]
where $v^T\eta v=\eta(v,v)=0$ because $v$ is null.
Bounding $|v|^2$ in terms of $|\Lambda v|^2$:
$|v|^2\leq |\Lambda v|^2/\sigma_{\min}(\Lambda)^2$. Combining,
$\theta':=\|\Lambda^T\eta\Lambda-\eta\|/\sigma_{\min}(\Lambda)^2$.

\emph{Bounds on $\|\Lambda\|$ and $\sigma_{\min}(\Lambda)$.}
We must establish both. Elements of $O(1,d-1)$ can in
general have arbitrarily large condition number (Lorentz
boosts), so $\sigma_{\min}(\Lambda)$ is not automatic from
$\Lambda^T\eta\Lambda\approx\eta$; we need the
Euclidean-size hypothesis on the anchors.

(i) \emph{Upper bound on $\|\Lambda\|$.} The boundedness
hypothesis~\eqref{eq:proc-bdd} gives
$\|B\|\leq\sqrt{d}\,C_{\mathrm{anch}}\,\alpha\ell$
(operator norm of a $d\times d$ matrix is bounded by $\sqrt d$
times the maximum column norm in the Euclidean inner
product). Combined with $\|A^{-1}\|=O(1/(\alpha\ell))$
from (C3),
\[
   \|\Lambda\|=\|BA^{-1}\|\leq\|B\|\cdot\|A^{-1}\|
   \leq C_d^{(1)}
\]
for a dimensional constant $C_d^{(1)}$.

(ii) \emph{Lower bound on $\sigma_{\min}(\Lambda)$ via the
determinant.} Taking determinants of
$\Lambda^T\eta\Lambda=\eta+R$ with $\|R\|\leq\delta=
O(\varepsilon/(\alpha\ell)^2)$:
\[
   \det(\Lambda)^2\cdot\det(\eta)
   =\det(\eta+R)=\det(\eta)(1+O(\delta))
\]
(expanding the determinant of a perturbation; the linear
term is $\det(\eta)\,\mathrm{tr}(\eta^{-1}R)=O(\delta)$,
higher-order terms are $O(\delta^2)$). Since $\det(\eta)=(-1)$,
this gives $\det(\Lambda)^2=1+O(\delta)$, hence
$|\det(\Lambda)|\geq 1/2$ for $\delta$ sufficiently small.
Now from the SVD, $|\det(\Lambda)|=\prod_{i=1}^d\sigma_i(\Lambda)$,
so
\[
   \sigma_{\min}(\Lambda)
   =\frac{|\det(\Lambda)|}{\prod_{i<d}\sigma_i(\Lambda)}
   \geq\frac{|\det(\Lambda)|}{\|\Lambda\|^{d-1}}
   \geq\frac{1/2}{(C_d^{(1)})^{d-1}}=:c_d>0,
\]
a dimensional lower bound.

Hence $\Lambda$ is uniformly conditioned,
$\kappa(\Lambda)=\|\Lambda\|/\sigma_{\min}(\Lambda)\leq C_d^{(1)}/c_d=:K$,
and $\theta':=\delta/c_d^2=O(\varepsilon/(\alpha\ell)^2)$. The
uniform-conditioning hypothesis of Lemma~\ref{lem:approx-az} is met,
so the lemma yields $\Omega>0$ and $\hat\Lambda\in O(1,d-1)$ with
$\|\Lambda-\Omega\hat\Lambda\|\leq C\,\|\Lambda\|\,\theta'$. Here
$\Lambda$ approximately preserves the \emph{full} Minkowski form,
$\|\Lambda^T\eta\Lambda-\eta\|\leq\delta$, so the conformal factor is
pinned to unity: $\Omega^2=-(\Lambda^T\eta\Lambda)_{00}=1+O(\delta)$.
Therefore
\[
   \|\Lambda-\hat\Lambda\|
   \leq\|\Lambda-\Omega\hat\Lambda\|+|\Omega-1|\,\|\hat\Lambda\|
   = O(\varepsilon/(\alpha\ell)^2),
\]
using $\|\Lambda\|=O(1)$ and $\|\hat\Lambda\|=O(1)$.

\emph{Step 4 (residuals).}
For each $m$:
\[
   \|b_m-\hat\Lambda a_m\|
   = \|(\Lambda-\hat\Lambda)a_m\|
   \leq \|\Lambda-\hat\Lambda\|\cdot\|a_m\|
   = O\bigl(\varepsilon/(\alpha\ell)^2\bigr)\cdot O(\alpha\ell)
   = O\bigl(\varepsilon/(\alpha\ell)\bigr).
\]
Substituting $\varepsilon=\varepsilon_\tau(\alpha\ell)^2$
gives $\|b_m-\hat\Lambda a_m\|=O(\varepsilon_\tau\alpha\ell)$.
\end{proof}

\section{Proof of the Approximate Alexandrov--Zeeman lemma}
\label{app:approx-az}

\begin{proof}[Proof of Lemma~\ref{lem:approx-az}]
\textit{Step~1: $A^T\eta A$ is approximately proportional to $\eta$.} Define the symmetric quadratic form $Q:=A^T\eta A$. The hypothesis on $A$ gives, for every null $v\in\mathcal{N}$,
\[
   |v^TQv| = |\eta(Av,Av)| \leq \theta|Av|^2
   \leq \theta\|A\|^2|v|^2.
\]

We claim (Lemma~\ref{lem:quadform-stability} in Appendix~\ref{app:quadform}) that any symmetric form $Q$ on $\R^d$ satisfying $|v^TQv|\leq\beta|v|^2$ for all null $v$ is close to a multiple of $\eta$:
\begin{equation}\label{eq:quadform-stability}
   Q = \alpha\eta + R,\qquad \|R\|\leq C_d\,\beta,
\end{equation}
for some scalar $\alpha\in\R$ and a dimensional constant $C_d$. The idea: parametrize null vectors as $v=(1,\hat n)$ with $\hat n\in S^{d-2}$ and use $\hat n\mapsto-\hat n$ symmetry together with polarization on $\hat n\in S^{d-2}$ to extract $\alpha:=-Q_{00}$ and bound the residual entries of $R:=Q-\alpha\eta$. Full calculation in Appendix~\ref{app:quadform}.

Applying this to our $Q=A^T\eta A$ with $\beta=\theta\|A\|^2$:
\begin{equation}\label{eq:Q-near-Omega-eta}
   A^T\eta A = \Omega^2\eta + R,\qquad
   \|R\|\leq C_d\,\theta\,\|A\|^2,
\end{equation}
where we have written $\alpha=\Omega^2$ for some real $\Omega$ (the sign of $\alpha$ will be confirmed positive below, since $A$ is invertible and the cone condition forces $\alpha\neq 0$).

\textit{Step~2: $\Omega>0$ and $A$ is close to $\Omega$ times a Lorentz transformation.}

\emph{Lower bound on $|\Omega^2|$.} Because $A$ is invertible, multiply~\eqref{eq:Q-near-Omega-eta} on the right by $A^{-1}$ and on the left by $(A^{-1})^T$ to isolate $\eta$:
\begin{equation}\label{eq:Omega-isolate}
   \Omega^2\,(A^{-1})^T\eta\,A^{-1} = \eta - (A^{-1})^T R\,A^{-1}.
\end{equation}
Take operator norms. On the left, $\|(A^{-1})^T\eta\,A^{-1}\|\leq\|A^{-1}\|^2$ (using $\|\eta\|=1$); on the right, the reverse triangle inequality gives $\|\eta-(A^{-1})^TRA^{-1}\|\geq 1-\|A^{-1}\|^2\|R\|$. Hence, substituting $\|R\|\leq C_d\theta\|A\|^2$ from~\eqref{eq:Q-near-Omega-eta} and the conditioning bound $\kappa(A)=\|A\|\,\|A^{-1}\|\leq K$,
\begin{equation}\label{eq:Omega-lower}
   |\Omega^2|\,\|A^{-1}\|^2 \;\geq\; 1-\|A^{-1}\|^2\|R\|
   \;\geq\; 1-K^2 C_d\,\theta.
\end{equation}
Choosing the threshold $\theta_0\leq 1/(2C_dK^2)$ (and shrinking it further below so that the Step~3 tube constraint also holds), the hypothesis $\theta\leq\theta_0$ makes the right side at least $\tfrac12$, so
\begin{equation}\label{eq:Omega-mag}
   |\Omega^2| \;\geq\; \frac{1}{2\|A^{-1}\|^2}
   \;\geq\; \frac{\|A\|^2}{2K^2} \;>\;0 .
\end{equation}

\emph{Sign of $\Omega^2$.} Take any timelike vector $v_0$ (so $\eta(v_0,v_0)<0$). By~\eqref{eq:Q-near-Omega-eta}, $\eta(Av_0,Av_0)=\Omega^2\eta(v_0,v_0)+v_0^TRv_0$, and the left side is strictly negative: since $A$ is an invertible linear map, hence a homeomorphism, the division of $\R^d\setminus\mathcal N$ into three connected components ensures $A$ cannot map the timelike interior to the single spacelike exterior. If $\Omega^2$ were negative then, taking $v_0$ of unit norm, $\Omega^2\eta(v_0,v_0)\geq|\Omega^2|\geq\tfrac12\|A\|^2/K^2$ by~\eqref{eq:Omega-mag}, while $|v_0^TRv_0|\leq\|R\|\leq C_d\theta\|A\|^2<\tfrac12\|A\|^2/K^2$ by the threshold; this would force $\eta(Av_0,Av_0)>0$, a contradiction. Hence $\Omega^2>0$, and combining with~\eqref{eq:Omega-mag},
\begin{equation}\label{eq:ratio-bound}
   \frac{\|A\|^2}{\Omega^2} \;\leq\; 2K^2 .
\end{equation}

Define $B:=\Omega^{-1}A$. Then $B^T\eta B=\Omega^{-2}A^T\eta A=\eta+\Omega^{-2}R$, so by~\eqref{eq:ratio-bound}
\begin{equation}\label{eq:B-defect}
   \|B^T\eta B-\eta\| = \Omega^{-2}\|R\|
   \leq \frac{C_d\,\theta\,\|A\|^2}{\Omega^2}
   \leq 2C_d\,K^2\,\theta .
\end{equation}

\textit{Step~3: Stability of the Lorentz group.} We use the implicit function theorem (in submersion form) to conclude that any matrix $B$ with $B^T\eta B$ close to $\eta$ must be close to some $\Lambda\in O(1,d-1)$. The relevant map is
\[
   \Psi\colon GL(d,\R)\to\mathrm{Sym}(d),
   \qquad
   \Psi(B):=B^T\eta B,
\]
with $\Psi^{-1}(\eta)=O(1,d-1)$. Its differential at any $\Lambda\in O(1,d-1)$ is, by the product rule applied to the bilinear form $B\mapsto B^T\eta B$:
\[
   d\Psi_\Lambda(E)=\Lambda^T\eta E+E^T\eta\Lambda
   \in\mathrm{Sym}(d),
\]
where $E\in\R^{d\times d}=T_\Lambda GL(d,\R)$ is any tangent direction.

\emph{Surjectivity of $d\Psi_\Lambda$ onto $\mathrm{Sym}(d)$.}
Given any $S\in\mathrm{Sym}(d)$, we show there exists an $E$ with $d\Psi_\Lambda(E)=S$. Set $E:=(\eta\Lambda)^{-T}(S/2)$ (the inverse exists since $\eta$ and $\Lambda$ are both invertible). Then
\begin{align*}
   d\Psi_\Lambda(E) &=\Lambda^T\eta\cdot(\eta\Lambda)^{-T}(S/2) +\bigl((\eta\Lambda)^{-T}(S/2)\bigr)^T\eta\Lambda\\
   &=(\eta\Lambda)^T(\eta\Lambda)^{-T}(S/2) +(S/2)^T(\eta\Lambda)^{-1}(\eta\Lambda)\\
   &=S/2+S^T/2=S
\end{align*}
since $S$ is symmetric. Hence $d\Psi_\Lambda$ is surjective onto $\mathrm{Sym}(d)$. In particular $\Psi$ is a submersion at every point of $O(1,d-1)$.

\emph{Consequence: quantitative tubular structure.}
Because $O(1,d-1)$ is noncompact, $\Psi$ does \emph{not} admit a tubular neighborhood of uniform radius: the right inverse of $d\Psi_\Lambda$ exhibited above has norm $\sim\|\Lambda^{-1}\|$, which is unbounded along the boost directions, so both the tube radius and the projection's Lipschitz constant degrade as $\|\Lambda^{-1}\|\to\infty$. Quantitative control is restored by confining $B$ to a \emph{compact} subset of $GL(d,\R)$, which the two standing hypotheses already force. Bounded condition number alone does not suffice, since $B=nI$ has $\kappa(B)=1$ for every $n$; it is the form bound that supplies the missing scale. Fix $\delta_\star:=\tfrac12$ and suppose $\kappa(B)\leq K$ and $\|B^T\eta B-\eta\|\leq\delta_\star$. By Weyl's inequality the eigenvalues of $B^T\eta B$ lie within $\delta_\star$ of those of $\eta$ (namely $\pm1$), so
\[
   \Bigl(\textstyle\prod_{i}\sigma_i(B)\Bigr)^2=\bigl|\det(B^T\eta B)\bigr|
   \in\bigl[(1-\delta_\star)^d,(1+\delta_\star)^d\bigr].
\]
Together with $\sigma_{\max}(B)\leq K\,\sigma_{\min}(B)$, the elementary inequalities $\prod_i\sigma_i\leq K^{d-1}\sigma_{\min}^d$ and $\prod_i\sigma_i\geq K^{-(d-1)}\sigma_{\max}^d$ pin every singular value of $B$ into the fixed interval
\[
   c_K:=\Bigl(\tfrac{(1-\delta_\star)^{d/2}}{K^{d-1}}\Bigr)^{1/d}
   \;\leq\;\sigma_i(B)\;\leq\;
   \bigl(K^{d-1}(1+\delta_\star)^{d/2}\bigr)^{1/d}=:C_K .
\]
Hence all such $B$ lie in the compact set $\mathcal K_K:=\{B:c_K\leq\sigma_{\min}(B)\leq\sigma_{\max}(B)\leq C_K\}\subset GL(d,\R)$, on which the right inverse of $d\Psi_B$ has norm $\leq 1/c_K$. The submersion form of the implicit function theorem (\cite{KrantzParks2002}, Theorem~6.2.4), whose tube radius and projection Lipschitz constant depend only on bounds for $d\Psi_B$ and its right inverse, therefore applies with constants uniform over $\mathcal K_K$: there exist $\delta_0\in(0,\delta_\star]$ and $C>0$, depending only on $d$ and $K$, such that for any $B\in GL(d,\R)$ with $\kappa(B)\leq K$ and $\|B^T\eta B-\eta\|\leq\delta\leq\delta_0$, there exists $\Lambda\in O(1,d-1)$ with
\[
   \|B-\Lambda\|\leq C\,\delta.
\]

We apply this with $B=\Omega^{-1}A$. The condition number is scale-invariant, $\kappa(B)=\kappa(\Omega^{-1}A)=\kappa(A)\leq K$, and $\delta=\|B^T\eta B-\eta\|\leq 2C_d K^2\theta$ by~\eqref{eq:B-defect}. Shrinking $\theta_0$ if necessary so that $2C_dK^2\theta_0\leq\delta_0$, the hypothesis $\theta\leq\theta_0$ guarantees $\delta\leq\delta_0$; the tubular bound then yields
\[
   \|\Omega^{-1}A-\Lambda\|=\|B-\Lambda\|\leq C\,\delta = O(\theta),
\]
with $C$ depending only on $d$ and $K$. Finally, evaluating~\eqref{eq:Q-near-Omega-eta} on a unit timelike vector gives $\Omega^2\leq\|A\|^2+\|R\|\leq 2\|A\|^2$, so $\Omega\leq\sqrt2\,\|A\|$ and
\[
   \|A-\Omega\Lambda\|=\Omega\,\|B-\Lambda\|\leq C\,\|A\|\,\theta ,
\]
with $C$ depending only on $d$ and $K$.
\end{proof}

\section{Quadratic-form stability for null cones}
\label{app:quadform}

We prove the following stability statement used in
Step~1 of the proof of Lemma~\ref{lem:approx-az}.

\begin{lemma}[Quadratic-form stability]
\label{lem:quadform-stability}
Let $Q$ be a symmetric bilinear form on $\R^d$ satisfying $|v^TQv|\leq\beta|v|^2$ for every null vector $v$ (i.e.\ $\eta(v,v)=0$). Note that $|v|^2$ denotes the standard Euclidean norm squared. Then there exists a scalar $\alpha\in\R$ and a residual symmetric form $R$ with
\[
   Q = \alpha\eta + R,\qquad \|R\|\leq C_d\,\beta,
\]
where $C_d$ depends only on $d$.
\end{lemma}
\begin{proof}
Parametrize null vectors as $v=(1,\hat n)$ with $\hat n\in S^{d-2}$ (using $\eta=\mathrm{diag}(-1,+1,\ldots,+1)$). Then
\[
   v^TQv = Q_{00} + 2Q_{0i}\hat n^i + Q_{ij}\hat n^i\hat n^j.
\]
Replacing $\hat n\to-\hat n$ (also null) gives
\[
   Q_{00}-2Q_{0i}\hat n^i+Q_{ij}\hat n^i\hat n^j.
\]
Subtracting: $|4Q_{0i}\hat n^i|\leq 4\beta$ for all
$\hat n\in S^{d-2}$, hence $|Q_{0i}|\leq\beta$ for each
spatial index $i$ (achieved by choosing $\hat n=e_i$).
Adding the two: $|2Q_{00}+2Q_{ij}\hat n^i\hat n^j|\leq 4\beta$,
so
\[
   |Q_{ij}\hat n^i\hat n^j+Q_{00}|\leq 2\beta
   \quad\text{for every unit }\hat n\in S^{d-2}.
\]
This says the spatial quadratic form $\hat n\mapsto
Q_{ij}\hat n^i\hat n^j$ is approximately constant on the
sphere, namely close to the constant $-Q_{00}$.

Set $\alpha:=-Q_{00}$ (so that $\alpha\eta_{00}=Q_{00}$
since $\eta_{00}=-1$). The residual $R:=Q-\alpha\eta$ has
\begin{itemize}
\item $R_{00}=Q_{00}-\alpha\eta_{00}=Q_{00}-(-Q_{00})(-1)=0$
exactly;
\item $|R_{0i}|=|Q_{0i}|\leq\beta$ (since $\eta_{0i}=0$);
\item $R_{ij}=Q_{ij}-\alpha\eta_{ij}=Q_{ij}+Q_{00}\delta_{ij}$
satisfies $|R_{ij}\hat n^i\hat n^j|\leq 2\beta$ for all
unit $\hat n$.
\end{itemize}
The diagonal entries of $R_{ij}$ are bounded by setting $\hat n=e_i$, yielding $|R_{ii}|\leq 2\beta$. The off-diagonal entries are recovered via polarization $2R_{ij}=R(\hat n_+,\hat n_+)-R(\hat n_-,\hat n_-)$ with $\hat n_\pm=(e_i\pm e_j)/\sqrt{2}$. Applying our spherical bound to these unit vectors gives $|2R_{ij}|\leq 2(2\beta)$, hence $|R_{ij}|\leq 2\beta$. Since every individual entry of the matrix $R$ is strictly bounded by a multiple of $\beta$, the overall operator norm satisfies $\|R \| \leq C_d\,\beta$. 
\end{proof}

\section{Anchor configuration: explicit calculations}
\label{app:anchor-checks}

We verify properties (C1)--(C3) of
Lemma~\ref{lem:target-config} for the explicit anchor
configuration~\eqref{eq:anchor-config}.

\begin{proof}
(C1): For $m=d$, $\eta(v_d^*,v_d^*)=-(16\alpha\ell)^2$. For
$m\in\{1,\ldots,d-1\}$,
$\eta(v_m^*,v_m^*)=-(16\alpha\ell)^2+(2\alpha\ell)^2
=-(\alpha\ell)^2(256-4)<0$.

(C2): Fix $p$ with $|p|\leq 2\alpha\ell$ (the full active
ball). For $m=0$:
$a_0^*-p=-8\alpha\ell\,e_0-p$, so
$\eta(a_0^*-p,a_0^*-p)=-(8\alpha\ell+p^0)^2+|\vec p|^2$.
Since $|p^0|,|\vec p|\leq 2\alpha\ell$, $(8\alpha\ell+p^0)^2
\geq(8\alpha\ell-2\alpha\ell)^2=(6\alpha\ell)^2=36(\alpha\ell)^2$,
and $|\vec p|^2\leq 4(\alpha\ell)^2$. Since $36>4$, the
displacement is chronological. Symmetrically for $m=d$
(with $+8\alpha\ell$).

For $m\in\{1,\ldots,d-1\}$:
$a_m^*-p=(8\alpha\ell-p^0)e_0+(2\alpha\ell-p^m)e_m
-\sum_{j\neq m}p^j e_j$.
The time component squared is $(8\alpha\ell-p^0)^2
\geq(6\alpha\ell)^2=36(\alpha\ell)^2$ (worst case
$p^0=2\alpha\ell$); the spatial-component squared is
$(2\alpha\ell-p^m)^2+\sum_{j\neq m}(p^j)^2
\leq(2\alpha\ell+2\alpha\ell)^2+(2\alpha\ell)^2
=(16+4)(\alpha\ell)^2=20(\alpha\ell)^2$ (worst case
$|p^m|=2\alpha\ell$ with sign opposing, and the remaining
spatial budget $\leq 2\alpha\ell$). The chronological
condition $36>20$ holds with margin, independent of
$\alpha\geq 16$ and of dimension.

(C3): The matrix $V^*$ has the explicit form
\[
   V^*
   =\begin{pmatrix}
   16\alpha\ell & 2\alpha\ell & 0 & \cdots & 0 \\
   16\alpha\ell & 0 & 2\alpha\ell & \cdots & 0 \\
   \vdots & & & \ddots & \\
   16\alpha\ell & 0 & 0 & \cdots & 2\alpha\ell \\
   16\alpha\ell & 0 & 0 & \cdots & 0
   \end{pmatrix}
\]
(rows $v_1^*,\ldots,v_{d-1}^*,v_d^*$, columns $e_0,e_1,\ldots,e_{d-1}$).
Subtracting the last row from the first $d-1$ leaves
$2\alpha\ell\,I_{d-1}$ in the spatial block, so $V^*$ is
invertible with $\|(V^*)^{-1}\|=O(1/(\alpha\ell))$.
Writing $V^*=\alpha\ell\,M_0$ with $M_0$ the fixed integer
matrix of entries displayed above, the scale $\alpha\ell$
cancels from $\kappa(V^*)=\kappa(M_0)$, so the condition
number depends only on $d$ and is independent of $\alpha$
and $\ell$.
\end{proof}

\section{Poisson faithful embedding: continuity and rate-scaling}
\label{app:poisson-faithful-details}

We complete the proof of Lemma~\ref{lem:poisson-faithful}
by verifying the continuity between net points and computing
the F2 tolerance at the relevant scales.

\emph{Continuity between net points.}
We must show that if F2 holds at every net point, it holds
(with a slightly enlarged tolerance) at every diamond in the
admissible range. The net spacing must be \emph{scale-dependent}:
a uniform spacing $\tau_{\min}$ does not suffice, because the
count change from displacing the endpoints by $\eta$ is
$\rho\,C\tau^{d-1}\eta$, which relative to the F2 tolerance
$\delta_D\,\rho\,\mathrm{Vol}_g(D)\asymp\delta_\tau\rho\tau^d$ is
$O\bigl(\eta/(\tau\delta_\tau)\bigr)$; with $\delta_\tau\propto
\tau^{-d/2}$ this ratio is worst at the \emph{coarse} end
$\tau=c_*\lambda$, and a fixed spacing fails there. We therefore
net each scale separately.

Partition the admissible range $[\tau_{\min},c_*\lambda]$ into
$O(\log\rho V_M)$ dyadic bands $\tau\in[\sigma/2,\sigma]$. On
each band place a net $\mathcal N_\sigma$ of endpoint pairs with
spacing
\[
   \eta(\sigma):=c\,\sigma\,\delta_\sigma,\qquad
   \delta_\sigma:=K_d\sqrt{\log\rho V_M/(\rho\kappa_d\sigma^d)},
\]
for a fixed small $c\in(0,1)$. Within the admissible range
$\tau\leq c_*\lambda$ the metric is $g=\eta+O(\tau^2/\lambda^2)$
in normal coordinates at the diamond center, so the Alexandrov
volume $\mathrm{Vol}_g(D(p,q))$ is smooth in $(p,q)$ with
$|\partial_p\mathrm{Vol}_g(D)|\leq C\,\mathrm{Vol}_g(\partial D)
\leq C\tau^{d-1}$ (divergence-theorem identity at sub-curvature
scales). Hence for $(p,q)$ in the band with nearest net point
$(p^*,q^*)$, $d_h(p,p^*)+d_h(q,q^*)\leq\eta(\sigma)$,
\[
   |\mathrm{Vol}_g(D(p,q))-\mathrm{Vol}_g(D(p^*,q^*))|
   \leq C\tau^{d-1}\eta(\sigma)
   \leq cC\,\delta_\sigma\,\rho^{-1}\!\cdot\!\rho\,\sigma^{d}
   \asymp cC\,\delta_D\,\mathrm{Vol}_g(D),
\]
using $\tau\asymp\sigma$. The count $N_D$ changes only by points
in the symmetric difference, of volume $\leq C\tau^{d-1}\eta(\sigma)$,
whose Poisson count has mean $\leq cC\,\delta_D\,\rho\,\mathrm{Vol}_g(D)$
and is controlled by the same Chernoff bound. Choosing $c$ small,
both the deterministic volume change and the random count change
stay within the F2 tolerance, so F2 transfers from $\mathcal N_\sigma$
to every diamond in the band, the tolerance enlarged by the fixed
factor $1+O(c)$ absorbed into $K_d$.

The total net $\bigcup_\sigma\mathcal N_\sigma$ has cardinality
$\sum_\sigma O((V_M/\eta(\sigma)^d)^2)$. Since $\eta(\sigma)=c\sigma\delta_\sigma$
with $\delta_\sigma\geq\delta_{c_*\lambda}$ a fixed power of $\rho$, each
$\eta(\sigma)^{-d}$ is at most a fixed power of $\rho$, so the total is a fixed
power $(\rho V_M)^{O(d)}$ times the $O(\log\rho V_M)$ bands. Taking the Chernoff
tail exponent for each net point larger
than that power (the same tunable-tail device as in
Proposition~\ref{prop:longest-chain}) makes the union bound over
the entire net close, giving F2 simultaneously at all admissible
diamonds with probability $\geq 1-(\rho V_M)^{-K_d'}$.

\emph{Scale separation.} Follows directly from the
hypothesis $\rho\lambda^d\geq c_*^{-2d}(\log\rho V_M)^2$.

\emph{Scaling at relevant scales.}
At the smoothing scale $\ell$ (used throughout the construction
of $\Phi$),
$\delta_\ell=K_d\sqrt{\log\rho V_M/(\rho\ell^d)}$,
which vanishes as a power of $\rho$. The same holds at the scaffold
scale $\alpha\ell$ and at the finest admissible scale $\tau_{\min}$,
$\delta_{\tau_{\min}}=K_d\sqrt{\log\rho V_M/(\rho\tau_{\min}^d)}$. Thus the
scale-dependent $\delta_D$ preserves the polynomial
convergence rate of the main theorem at every scale where
F2 is invoked.

\bibliographystyle{plain}
\bibliography{refs}

\end{document}